\documentclass[11pt]{article}

\usepackage[]{graphicx,float,latexsym,times}
\usepackage{amsfonts,amsmath,amstext,amssymb,amsthm}
\usepackage{caption2}

\long\def\symbolfootnote[#1]#2{\begingroup%
\def\thefootnote{\fnsymbol{footnote}}\footnote[#1]{#2}\endgroup}

\usepackage{titlesec} 

\titleformat{\section}{\large\bfseries}{\thesection.}{.5em}{}
\titlespacing*{\section}{0pt}{*3}{*2}
\titleformat{\subsection}{\normalfont\bfseries}{\thesubsection.}{.5em}{}
\titlespacing*{\subsection} {0pt}{*3}{*2}
\titleformat{\subsubsection}{\normalfont\bfseries}{\thesubsubsection.}{.5em}{}
\titlespacing*{\subsubsection} {0pt}{*3}{*2}

\theoremstyle{plain} 
\newtheorem{theorem}{Theorem}[section]
\newtheorem{lemma}{Lemma}[section]

\theoremstyle{definition} 

\newtheorem{remark}{Remark}[section]
\newtheorem{prop}{Proposition}[section]

\numberwithin{equation}{section} 

\begin{document}

\title{\textbf{\Large Normal Power Series Class of Distributions: Model, Properties and Applications}}


\maketitle


\author{
\begin{center}
\vskip -1cm

\textbf{\large Eisa Mahmoudi and Hamed Mahmoodian}

Department of Statistics, Yazd University, Yazd, Iran  
\end{center}
}

\noindent\symbolfootnote[0]{\normalsize Address correspondence to E. Mahmoudi,
Department of Statistics, Yazd University, P.O. Box 89175-741, Yazd, Iran; Fax: +98-35-38210695; E-mail: emahmoudi@yazd.ac.ir}

{\small \noindent\textbf{Abstract:} A new class of distributions, called as normal power series (NPS), which contains the normal one as a particular case, is introduced in this paper. This new class which is obtained by compounding the normal and power series distributions, is presented as an alternative to the class of skew-normal and Balakrishnan skew-normal distributions, among others. The
density and distribution functions of this new class of distributions, are given by a closed expression which allows us to easily compute probabilities, moments and related measurements.
Estimation of the parameters of this new model by maximum likelihood method via an EM-
algorithm is given. Finally, some applications are shown as examples.}
\\ \\
{\small \noindent\textbf{Subject Classifications:} }

\section{INTRODUCTION} \label{s:Intro}

Recently, many distributions to model lifetime data have been studied and
generalized by compounding of some discrete and important lifetime
distributions. Adamidis and Loukas (1998), introduced exponential-geometric
(EG) distribution by compounding the exponential and geometric
distributions. In the similar manner exponential-Poisson (EP),
exponential-logarithmic (EL), exponential-power series (EPS),
Weibull-geometric (WG) and Weibull-power series (WPS) distributions were
introduced by Tahmasbi and Rezaei (2008), Chahkandi and Ganjali (2009),
Barreto-Souza et al. (2011) and Morais and Barreto-Souza (2011),
respectively.

In recent years, techniques for extending the family of normal distributions have been proposed. The method applied here can be considered an alternative to the well-known skew-normal distribution (Azzalini (1985)), whose properties (Azzalini (1986); Azzalini and Chiogna (2004)), estimation (Gupta and Gupta (2008)), diagnostics (Xie et al. (2009)), generalization (Gupta and Gupta (2004) and multivariate extension (Azzalini and Valle, 1996; Azzalini and Capitanio, 1999; Arnold and Beaver,
2002) have been widely developed. Other ways of obtaining skewed normal distributions have also been introduced, such as the Balakrishnan skew-normal density in Sharafi and Behboodian (2008), the variance-gamma process in Fung and Seneta (2007) and the generalized normal distribution in Nadarajah (2005), among others. Whenever the Fisher information matrix of this skew-normal model is singular for values of the added parameter $\lambda$, and the maximal likelihood estimate of this parameter can be infinite with a positive probability, an alternative model would be desirable.

In this paper, we introduce a new generalization of the normal distribution which is called the NPS class of distributions and is denoted by $NPS(\mu,\sigma, \theta)$. The method used to insert
the new parameter is described in Marshall and Olkin (1997) for the first time, where it was applied to the exponential and Weibull families. This method enables us to obtain explicit expressions for the probability density functions and survival functions and allow us to MLE-estimate via an EM-
algorithm of the model parameters with the help of appropriate R software.

Note that all of these distribution have ranges on $(0,\infty )$, but in this
paper we introduced the new generalization of normal distribution, which has range in $(-\infty
,\infty )$. To begin with, we shall use the following notation throughout this paper: 
$\phi (\cdot )$ for the standard normal probability density function (pdf), $
\phi _{n}(\cdot \ ;\mbox{\boldmath $\mu$},\mbox{\boldmath $\Sigma$})$ for
the pdf of $N_{n}(\mbox{\boldmath $\mu$},\mbox{\boldmath $\Sigma$})$ ($n$
-variate normal distribution with mean vector $\mbox{\boldmath $\mu$}$ and
covariance matrix $\mbox{\boldmath $\Sigma$}$), $\Phi _{n}(\cdot \ ;%
\mbox{\boldmath $\mu$},\mbox{\boldmath $\Sigma$})$ for the cdf of $N_{n}(%
\mbox{\boldmath $\mu$},\mbox{\boldmath $\Sigma$})$ (in both singular and
non-singular cases), simply $\Phi _{n}(\cdot \ ;\mbox{\boldmath
$\Sigma$})$ for the case when $\mbox{\boldmath $\mu$}=\mathbf{0.}$

The rest of this paper is organized as follows. In Section 2, we define the class of NPS distributions. The density, hazard
rate and survival functions and some of their properties are given in this section. In Section 3, we derive moments of NPS
distributions by two method. In Section 4, we present some special distributions which are studied in details. Some properties of sub model of NPS distristribution are studied in section 5. Estimation of the parameters by maximum likelihood method and inference for large sample are presented in Section 6. The EM-algorithm with a method evaluating the standard errors from the EM-algorithm is presented in Section 7. Simulation study is given in Section 8. Applications to two real data sets are given in Section 9. Finally, Section 10 concludes the paper.
\section{The NPS class}

Let $X_{1},..,X_{N}$ be a random sample from normal distribution with mean $
\mu $ and variance $\sigma ^{2}$ and let the random variable $N$ has a power series distributions
(truncated at zero) with the probability mass function
\begin{equation*}
P(N=n)=\frac{a_{n}\theta ^{n}}{C(\theta )},
\end{equation*}
where $a_{n}>0$ depends only on $n$, $C(\theta )=\sum_{n=1}^{\infty
}a_{n}\theta ^{n}$ and $\theta \in (0, s)$ ($s$ can be $\infty$) is such
that $C(\theta)$ is finite. Table 1 lists some particular cases of the truncated (at zero) power series distributions (geometric, Poisson, logarithmic, binomial and negative binomial). Detailed
properties of power series distributions can be found in Noack (1950). Here, $C'(\theta)$, $C''(\theta)$ and $C'''(\theta)$ denote the first, second and third derivatives of $C(\theta)$ with respect to $\theta$, respectively.
\begin{table}
\centering
\caption{Quantities for power series distributions}
\begin{tabular}{|l|llllll|}
\hline\\
Distribution&$a_n$&$C(\theta)$&$C'(\theta)$&$C''(\theta)$&$C'''(\theta)$&$s$\\
\hline\\
Geometric&1&$\theta(1-\theta)^{-1}$&$(1-\theta)^{-2}$&$2(1-\theta)^{-3}$&$6(1-\theta)^{-4}$&1\medskip\\
Poisson&$n!^{-1}$&$e^{\theta}-1$&$e^{\theta}$&$e^{\theta}$&$e^{\theta}$&$\infty$\medskip\\
Logarithmic&$n^{-1}$&$-\log(1-\theta)$&$(1-\theta)^{-1}$&$(1-\theta)^{-2}$&$2(1-\theta)^{-3}$&1\medskip\\ 
Binomial&${k \choose n}$&$(1+\theta)^k-1$&$\frac{k}{(1+\theta)^{1-k}}$&$\frac{k(k-1)}{(1+\theta)^{2-k}}$&$\frac{k(k-1)(k-2)}{(1+\theta)^{3-k}}$&$\infty$\medskip\\
Neg. Binomial&${n-1 \choose k-1}$&$\frac{\theta^k}{(1-\theta)^k}$&$\frac{k\theta^{k-1}}{(1-\theta)^{k+1}}$&$\frac{k(k+2\theta-1)}{\theta^{2-k}(1-\theta)^{k+2}}$&$\frac{k(k^2+6k\theta+6\theta^2-3k-6\theta+2)}{\theta^{3-k}(1-\theta)^{k+3}}$&1\\
\hline

\end{tabular}\label{tab1}
\end{table}

Let $Y=X_{(N)}=\max \left( X_{1},..,X_{N}\right),$ then the conditional cdf of 
$Y|N=n$ is given by%
\begin{equation*}
G_{Y|N=n}(y)=\left( \Phi \left( \frac{y-\mu }{\sigma }\right) \right) ^{n},
\end{equation*}%
where $\Phi (\cdot )$ denotes the cdf of standard normal distribution.\\
The cdf of normal power series (NPS) class of distributions is defined by the
marginal cdf of $Y$, i.e.,
\begin{equation}\label{NPScdf1}
F(y;\mu,\sigma,\theta)=\sum_{n=1}^{\infty }\frac{a_{n}\theta ^{n}}{C(\theta )}\left( \Phi
\left( \frac{y-\mu }{\sigma }\right) \right) ^{n}=\frac{C\left( \theta \Phi
\left( \frac{y-\mu }{\sigma }\right) \right) }{C(\theta )},\;y \in
\mathbb{R},~\mu \in \mathbb{R},~\sigma>0.
\end{equation}
We denote a random variable $Y$ follows NPS distributions by $NPS(\mu ,\sigma ,\theta )$. The density function of NPS follows immediately as:
\begin{equation}\label{NPSpdf}
f(y;\mu,\sigma,\theta)=\frac{\theta }{\sigma }\phi \left( \frac{y-\mu }{\sigma }\right) \frac{
C'(\theta \Phi \left( \frac{y-\mu }{\sigma }\right) )}{C(\theta )}. 
\end{equation}
The corresponding survival and hazard rate functions are 
\begin{equation*}
S(y;\mu,\sigma,\theta)=1-\frac{C(\theta \Phi \left( \frac{y-\mu }{\sigma }\right) )}{C(\theta )
}, 
\end{equation*}
and
\begin{equation*}
h(y;\mu,\sigma,\theta)=\frac{\theta }{\sigma }\frac{\phi \left( \frac{y-\mu }{\sigma }\right)
C'(\theta \Phi \left( \frac{y-\mu }{\sigma }\right) )}{C(\theta
)-C(\theta \Phi \left( \frac{y-\mu }{\sigma }\right) )}, 
\end{equation*}
respectively.\\
One can put $\mu=0$ and $\sigma=1$ and obtains the standard version on NPS class of distributions denote by $NPS(0,1,\theta)$. Fig. 1 shows the probability density function and hazard rate function of the classical normal distribution and the NPS distributions proposed in this
paper for some choices of $C(\theta)$. It can be seen that the new model is very versatile and that the value of $\theta$ has a substantial effect on the skewness of the probability density function.

\begin{prop}
For the pdf and cdf of NPS class of distributions, we have
\begin{equation}
\lim_{y\rightarrow \pm \infty }f(y;\mu,\sigma,\theta)=0,~~~~\lim_{y\rightarrow \pm \infty }h(y;\mu,\sigma,\theta)=\pm \infty. 
\end{equation}
\end{prop}

\begin{prop}
Let $c=\min \{n\in \mathbb{N\colon }a_{n}>0\}.$ As $\theta \rightarrow 0^{+}$ we have
\begin{eqnarray*}
\lim_{\theta \rightarrow 0^{+}}F(y;\mu,\sigma,\theta) &=&\lim_{\theta \rightarrow 0^{+}}\frac{%
C(\theta \Phi \left( \frac{y-\mu }{\sigma }\right) )}{C(\theta )}%
=\lim_{\theta \rightarrow 0^{+}}\frac{\sum_{n=1}^{\infty }a_{n}\theta
^{n}(\Phi \left( \frac{y-\mu }{\sigma }\right) )^{n}}{\sum_{n=1}^{\infty
}a_{n}\theta ^{n}} \\
&=&\lim_{\theta \rightarrow 0^{+}}\frac{\sum_{n=1}^{c-1}a_{n}\theta
^{n}(\Phi \left( \frac{y-\mu }{\sigma }\right) )^{n}+a_{c}\theta ^{c}(\Phi
\left( \frac{y-\mu }{\sigma }\right) )^{c}+\sum_{n=c+1}^{\infty }a_{n}\theta
^{n}(\Phi \left( \frac{y-\mu }{\sigma }\right) )^{n}}{\sum_{n=1}^{c-1}a_{n}%
\theta ^{n}+a_{c}\theta ^{c}+\sum_{n=c+1}^{\infty }a_{n}\theta ^{n}} \\
&=&\lim_{\theta \rightarrow 0^{+}}\frac{(\Phi \left( \frac{y-\mu }{\sigma }%
\right) )^{c}+a_{c}^{-1}\sum_{n=c+1}^{\infty }a_{n}\theta ^{n-c}(\Phi \left( 
\frac{y-\mu }{\sigma }\right) )^{n}}{1+a_{c}^{-1}\sum_{n=c+1}^{\infty
}a_{n}\theta ^{n}}=\left( \Phi \left( \frac{y-\mu }{\sigma }\right) \right)
^{c}.
\end{eqnarray*}
\end{prop}

\begin{prop}
The densities of NPS class of distributions can be written as infinite
number of linear combination of density of order statistics. We know that $
C'(\theta )=\sum_{n=1}^{\infty }na_{n}\theta ^{n-1}$, therefore
\begin{equation*}
f(y;\mu,\sigma,\theta)=\frac{\theta }{\sigma }\phi \left( \frac{y-\mu }{\sigma }\right) \frac{
C'(\theta \Phi \left( \frac{y-\mu }{\sigma }\right) )}{C(\theta )
}=\sum_{n=1}^{\infty }g_{X_{(n)}}(y;n)P(N=n),
\end{equation*}
where $g_{X_{(n)}}(y;n)$ denotes the density function of $X_{(n)}=\max \left(
X_{1},..,X_{n}\right).$
\end{prop}

\begin{prop}
The $\gamma $th quantile of the NPS class of
distributions is given by
\begin{equation*}
y_{\gamma }=\sigma \Phi ^{-1}\left( \frac{C^{-1}(\gamma C(\theta ))}{\theta }
\right) +\mu.
\end{equation*}
\end{prop}
\noindent One can use this expression for generating a random sample from NPS
distributions with generating data from uniform distribution.
\section{Moments of the NPS distributions}

In this section we give two method to obtain the moments of the NPS distributions.\\
\textbf{\textit{(i) First method}}\\
Jamalizadeh and Balakrishnan (2010) considered unified
skew-elliptical distributions which contains unified skew-normal distribution as a
special case. The univariate random variable $Z_{k,\mbox{\boldmath
$\theta$}}$, $\mbox{\boldmath
$\theta$}=\left( \mbox{\boldmath
$\lambda$},\mbox{\boldmath
$\gamma$},\mathbf{\Omega }\right) ,\mbox{\boldmath
$\lambda$},\mbox{\boldmath
$\gamma$}\in \mathbb{R}^{k}$ and $\mathbf{\Omega \in }\mathbb{R}^{k\times k}$
is a positive definite dispersion matrix, is said to have a unified
skew-normal distribution, denoted by $Z_{k,\mbox{\boldmath
$\theta$}}\sim SN\left( k,\mbox{\boldmath
$\theta$}\right) $, if its pdf is 
\begin{equation*}
\phi _{SN}\left( z;k,\mbox{\boldmath
$\theta$}\right) =\frac{\phi \left( z\right) \Phi _{k}\left( 
\mbox{\boldmath
$\lambda$}z+\mbox{\boldmath
$\gamma$};\mathbf{\Omega }\right) }{\Phi _{_{k}}\left( 
\mbox{\boldmath
$\gamma$};\mathbf{\Omega +}\mbox{\boldmath
$\lambda$}\mbox{\boldmath
$\lambda$}^{T}\right) }.
\end{equation*}%
Furthermore, the moment generating function of $Z_{k,%
\mbox{\boldmath
$\theta$}}\sim SN\left( k,\mbox{\boldmath
$\theta$}\right) $ is, for $s\in \mathbb{R}$, 
\begin{equation*}
M_{SN}\left( s;k,\mbox{\boldmath
$\theta$}\right) =\frac{\exp \left( \frac{1}{2}s^{2}\right) \Phi
_{_{k}}\left( \mbox{\boldmath
$\lambda$}s+\mbox{\boldmath
$\gamma$};\mathbf{\Omega +}\mbox{\boldmath
$\lambda$}\mbox{\boldmath
$\lambda$}^{T}\right) }{\Phi _{_{k}}\left( \mbox{\boldmath
$\gamma$};\mathbf{\Omega +}\mbox{\boldmath
$\lambda$}\mbox{\boldmath
$\lambda$}^{T}\right) }.
\end{equation*}
The mean of $Z_{k,\mbox{\boldmath
$\theta$}}$ can be obtained as the following lemma.

\begin{lemma}
If $Z_{k,\mbox{\boldmath
$\theta$}}\sim SN\left( k,\mbox{\boldmath
$\theta$}\right) $, then%
\begin{eqnarray*}
E\left( Z_{k,\mbox{\boldmath
$\theta$}}\right) &=&\frac{1}{\Phi _{_{k}}\left( \mbox{\boldmath
$\gamma$};\mathbf{\Omega +}\mbox{\boldmath
$\lambda$}\mbox{\boldmath
$\lambda$}^{T}\right) }\sum_{i=1}^{k}\frac{\lambda _{i}}{\sqrt{\omega
_{ii}+\lambda _{i}^{2}}}\phi \left( \frac{\gamma _{i}}{\sqrt{\omega
_{ii}+\lambda _{i}^{2}}}\right) \\
&&~\times \Phi _{k-1}\left( \mbox{\boldmath
$\gamma$}_{-i}-\frac{\gamma _{i}}{\omega _{ii}+\lambda _{i}^{2}}\left( 
\mbox{\boldmath
$\omega$}_{-ii}+\lambda _{i}\mbox{\boldmath
$\lambda$}_{-i}\right) ;\left( \mathbf{\Omega +}\mbox{\boldmath
$\lambda$}\mbox{\boldmath
$\lambda$}^{T}\right) _{-i\mid i}\right) ,
\end{eqnarray*}%
where, for some $i$, 
\begin{equation}
\mbox{\boldmath
$\lambda$}=\left( 
\begin{array}{c}
\lambda _{i} \\ 
\mbox{\boldmath
$\lambda$}_{-i}%
\end{array}%
\right) ,\ \mbox{\boldmath
$\gamma$}=\left( 
\begin{array}{c}
\gamma _{i} \\ 
\mbox{\boldmath
$\gamma$}_{-i}%
\end{array}%
\right) ,\ \mathbf{\Omega =}\left( 
\begin{array}{cc}
\omega _{ii} & \mbox{\boldmath
$\omega$}_{-ii}^{T} \\ 
\mbox{\boldmath
$\omega$}_{-ii} & \mathbf{\Omega }_{-i-i}%
\end{array}%
\right) ,  \notag
\end{equation}%
with$\left( \mathbf{\Omega +}\mbox{\boldmath
$\lambda$}\mbox{\boldmath
$\lambda$}^{T}\right) _{-i\mid i}=\mathbf{\Omega }_{-i-i}+%
\mbox{\boldmath
$\lambda$}_{-i}\mbox{\boldmath
$\lambda$}_{-i}^{T}-\frac{\left( \mbox{\boldmath
$\omega$}_{-ii}+\lambda _{i}\mbox{\boldmath
$\lambda$}_{-i}\right) \left( \mbox{\boldmath
$\omega$}_{-ii}+\lambda _{i}\mbox{\boldmath
$\lambda$}_{-i}\right) ^{T}}{\omega _{ii}+\lambda _{i}^{2}}.$
\end{lemma}
In the special case when $\mbox{\boldmath
$\gamma$}=\mathbf{0}$, the moments can be determined rather easily. If in
this case the generalized skew-normal is denoted by $Z_{k,%
\mbox{\boldmath
$\lambda$},\mathbf{\Omega }}$, we then have 
\begin{equation*}
E\left( Z_{k,\mbox{\boldmath
$\lambda$},\mathbf{\Omega }}\right) =\frac{1}{\Phi _{_{k}}\left( \mathbf{0};%
\mathbf{\Omega +}\mbox{\boldmath
$\lambda$}\mbox{\boldmath
$\lambda$}^{T}\right) \sqrt{2\pi }}\sum_{i=1}^{k}\frac{\lambda _{i}}{\sqrt{%
\omega _{ii}+\lambda _{i}^{2}}}\ \Phi _{k-1}\left( \mathbf{0};\left( \mathbf{%
\Omega +}\mbox{\boldmath
$\lambda$}\mbox{\boldmath
$\lambda$}^{T}\right) _{-i\mid i}\right) ,
\end{equation*}%
where $\left( \mathbf{\Omega +}\mbox{\boldmath
$\lambda$}\mbox{\boldmath
$\lambda$}^{T}\right) _{-i\mid i}$ is as given above. In this case, was
obtained a recurrence formula for the moments of $Z_{k,%
\mbox{\boldmath
$\lambda$},\mathbf{\Omega }}$ $\sim SN\left( k,\mbox{\boldmath
$\lambda$},\mathbf{\Omega }\right) $. For simplicity, they presented in the
following lemma this recurrence formula for the case when $\mathbf{%
\Omega }$ is the correlation matrix.

\begin{lemma} 
We have, for $m=1,2,\cdots ,$ 
\begin{eqnarray*}
E\left( Z_{k,\mbox{\boldmath
$\lambda$},\mathbf{\Omega }}^{m+1}\right) &=&mE\left( Z_{k,%
\mbox{\boldmath
$\lambda$},\mathbf{\Omega }}^{m-1}\right) \\
&&+\frac{1}{\sqrt{2\pi }\Phi _{_{k}}\left( \mathbf{0};\mathbf{\Omega +}%
\mbox{\boldmath
$\lambda$}\mbox{\boldmath
$\lambda$}^{T}\right) }\sum_{i=1}^{k}\frac{\lambda _{i}\Phi _{_{k-1}}\left( 
\mathbf{0};\mathbf{\Omega }_{-i\mid i}\mathbf{+}\mbox{\boldmath
$\lambda$}_{i}^{\ast }\mbox{\boldmath
$\lambda$}_{i}^{\ast ^{T}}\right) }{\left( 1+\lambda _{i}^{2}\right) ^{\frac{%
m+1}{2}}}E\left( Z_{k-1,\mbox{\boldmath
$\lambda$}_{i}^{\ast },\mathbf{\Omega }_{-i\mid i}}^{m}\right) ,
\end{eqnarray*}%
where $\mbox{\boldmath
$\lambda$}_{i}^{\ast }=\frac{\mbox{\boldmath
$\lambda$}_{-i}}{\sqrt{1+\lambda _{i}^{2}}}$.
\end{lemma}
In addition when $\mathbf{X}\sim N_{n}\left( \mu \mathbf{1}_{n}%
\mathbf{,\sigma }^{2}\{(1-\rho )\mathbf{I}_{n}+\rho \mathbf{1}_{n}\mathbf{1}%
_{n}^{T}\}\right) ,\ \mu \in \mathbb{R},\ \sigma >0,\ -\frac{1}{n-1}<\rho <1$, they
proved that
\begin{equation}\label{dis. order}
\frac{X_{(r)}-\mu }{\sigma }\sim SN(n-1,\mbox{\boldmath
$\theta$}),
\end{equation}%
where 
\begin{equation*}
\mbox{\boldmath
$\theta$}=\left( \sigma (1-\rho )\mathbf{J}_{n-1}\mathbf{,0,\sigma }%
^{2}(1-\rho )\{\mathbf{I}_{n-1}+\rho \mathbf{J}_{n-1}\mathbf{J}%
_{n-1}^{T}\}\right) .
\end{equation*}%
with $\mathbf{J}_{n-1}$ $=\left( \mathbf{1}_{r-1}^{T},-\mathbf{1}%
_{n-r}^{T}\right) ^{T}$\ and $\mathbf{I}_{n-1}\in \mathbb{R}^{n-1\times n-1}$.\\
We used the above lemmas and equation \eqref{dis. order}, without loss of generality, to obtain the moment generating function, $k$th moment and the first moment of $NPS(\mu ,\sigma ,\theta ),$ when $\mu =0$ and $\sigma =1$ in the following proposition.\\
\begin{prop}
If $Y \sim NPS(0 ,1 ,\theta )$, then the moment generating function, $k$th moment and mean of $Y$ are given by 
\begin{eqnarray*}
M_{Y}(t) &=&\sum_{n=1}^{\infty }\frac{a_{n}\theta ^{n}}{C(\theta )}%
M_{X_{(n)}}(t)=\exp \left( \frac{1}{2}t^{2}\right) \sum_{n=1}^{\infty }\frac{%
a_{n}\theta ^{n}}{C(\theta )}\times n\Phi
_{n-1}(1_{n-1}t;I_{n-1}+1_{n-1}1_{n-1}^{T}) \\
&=&\exp \left( \frac{1}{2}t^{2}\right) E(N\Phi
_{N-1}(1_{N-1}t;I_{N-1}+1_{N-1}1_{N-1}^{T})),
\end{eqnarray*}

\begin{eqnarray*}
E(Y^{k+1}) &=&\sum_{n=1}^{\infty }\frac{a_{n}\theta ^{n}}{C(\theta )}
kE(Y^{k-1})+\sum_{n=1}^{\infty }\frac{a_{n}\theta ^{n}}{C(\theta )}\times 
\frac{(n-1)\Phi _{n-2}\left( \mathbf{0;}I_{n-2}+\frac{1}{2}%
1_{n-2}1_{n-2}^{T}\right) }{2\sqrt{\pi }\Phi _{n-1}\left( 
\mathbf{0;}I_{n-1}+1_{n-1}1_{n-1}^{T}\right) }  \\
&~~\times & E\left( Z_{n-2,\frac{1}{\sqrt{2}}%
1_{n-2},I_{n-2}}^{k}\right),
\end{eqnarray*}
and
\begin{equation*}
E(Y)=\frac{1 }{2\sqrt{\pi }}%
\sum_{n=1}^{\infty }\frac{a_{n}\theta ^{n}}{C(\theta )}\times n(n-1)\Phi
_{n-2}\left( \mathbf{0;}I_{n-2}+\frac{1}{2}1_{n-2}1_{n-2}^{T}\right),
\end{equation*}%
respectively.
\end{prop} 
One can derive the second moment of NPS distributions as
\begin{eqnarray*}
E(Y^{2})=1+\frac{1}{4\sqrt{3}\pi}\sum_{n=1}^{\infty}\frac{a_{n}\theta^{n}%
}{C(\theta)}\times n(n-1)(n-3)\Phi_{n-3}\left(  \mathbf{0};I_{n-3}+\frac{1}%
{3}1_{n-3}1_{n-3}^{T}\right).
\end{eqnarray*}

\textbf{\textit{(ii) Second method}}\\
In the following proposition we present another formulation for calculate the $k$th moment
around the origin of the random variable $Y\sim NPS(\mu ,\sigma ,\theta ).$ First, we give two well-known relationship, which are necessary in the following proposition. If $\Phi(x;\mu,\sigma)$ denotes the cdf of $N(\mu,\sigma^2)$ distribution, then we have\\
\begin{eqnarray}\label{Erf}
\Phi(x;\mu,\sigma)=\frac{1}{2}\left[1+erf\left(\frac{x-\mu}{\sigma\sqrt{2}}\right)\right],
\end{eqnarray}
and
\begin{eqnarray}\label{Erf-1}
\Phi^{-1}(t;\mu,\sigma)=\mu+\sigma\sqrt{2}erf^{-1}(2t-1).
\end{eqnarray}

\begin{prop}
We have
\begin{eqnarray*}
E(Y^{k})=\int_{0}^{1}\left\{  \mu+\sqrt{2}\sigma\operatorname{erf}^{-1}\left(
\frac{2C^{-1}\left(  C(\theta)u\right)  }{\theta}-1\right)  \right\}^{k}du.
\end{eqnarray*}
\end{prop}
\begin{proof}
We know that $E(Y^{k})=\int_{-\infty }^{\infty }y^{k}f(y;\mu
,\sigma ,\theta )dy$. Then substituting $f(y;\mu
,\sigma ,\theta )$ from \eqref{NPSpdf} and change of variable $\Phi \left( y;\mu ,\sigma \right) =t$ gives
\[
E(Y^{k})=\frac{\theta }{C(\theta )}\int_{0}^{1}\left[\Phi ^{-1}\left( t;\mu
,\sigma \right) \right]^{k}C^{^{\prime }}(\theta t)dt.
\]%
Now, by changing the variable to $u=\frac{C(t\theta )}{C(\theta )},$ we have 
\[
E(Y^{k})=\int_{0}^{1}\left\{ \Phi ^{-1}\left( \frac{C^{-1}(C(\theta )u)}{%
\theta };\mu ,\sigma \right) \right\} ^{k}du, 
\]%
thus, the result follows from the Eq. \eqref{Erf-1}.
\end{proof}

Table 1 contains values for $\mu_k$ ($k=1,~ 2,~ 3,~4$) in a $NPS(0,1,\theta)$ distributions for
some values of $\theta$. Also, the variance and the skewness coefficient $Sk$ given by
$Sk=\mu_3/\sigma^{3}$, is included in Table 1.
%
%

\section{Special cases of NPS class of distributions}
In this section four important sub-models of NPS class of distributions are studied in details. These models are normal geometric (NG), normal Poisson (NP), normal logarithmic (NL) and normal binomial (NB) distributions.\\
\subsection{Normal geometric distribution}
Using Table \ref{tab1}, the NPS distributions contain NG distribution when $a_{n}=1$ and $C(\theta )=\frac{\theta }{1-\theta }(0<\theta <1)$. Using Eq. \eqref{NPScdf1}, the cdf of NG is given by 
\begin{equation}\label{NGcdf}
F(y;\mu,\sigma,\theta)=\frac{(1-\theta )\Phi (\frac{y-\mu }{\sigma })}{1-\theta \Phi (\frac{%
y-\mu }{\sigma })},%
\;y\mathbf{\in }\mathbb{R},\ \ \qquad \ 
\end{equation}
The pdf and hazard rate functions of NG distribution are
\begin{equation}\label{NGpdf}
f(y;\mu,\sigma,\theta)=\frac{(1-\theta )\phi (\frac{y-\mu }{\sigma })}{\sigma (1-\theta \Phi (%
\frac{y-\mu }{\sigma }))^{2}},
\end{equation}
and

\begin{equation*}
h(y;\mu,\sigma,\theta)=\frac{(1-\theta )\phi (\frac{y-\mu }{\sigma })}{\sigma (1-\Phi (\frac{%
y-\mu }{\sigma }))(1-\theta \Phi (\frac{y-\mu }{\sigma }))},
\end{equation*}%
respectively, where $y \in \mathbb{R}$, $\mu \in \mathbb{R}$, $\sigma>0$ and $0<\theta<1$. We use the notation $Y \sim NG(\mu,\sigma,\theta)$ when the random variable $Y$ has NG distribution with location $\mu$, scale $\sigma$ and shape parameter $\theta$. 
\begin{remark}
Even when $\theta\leq0$, Equation \eqref{NGpdf} is also a density function.
We can then define the NG distribution by Equation \eqref{NGpdf}
for any $\theta<1$. Some special sub-models of the NG distribution
are obtained as follows. If $\theta = 0$, we have the
normal distribution. When $\theta\rightarrow 1^{-}$, the NG distribution
tends to a distribution degenerate in zero. Hence, the parameter $\theta$
can be interpreted as a concentration parameter. 
\end{remark}
Figure \ref{phNG} shows the NG density and hazard rate functions for selected values $\theta$ where $\mu=0$ and $\sigma=1$.

\begin{figure}
\centering
\includegraphics[scale=0.40]{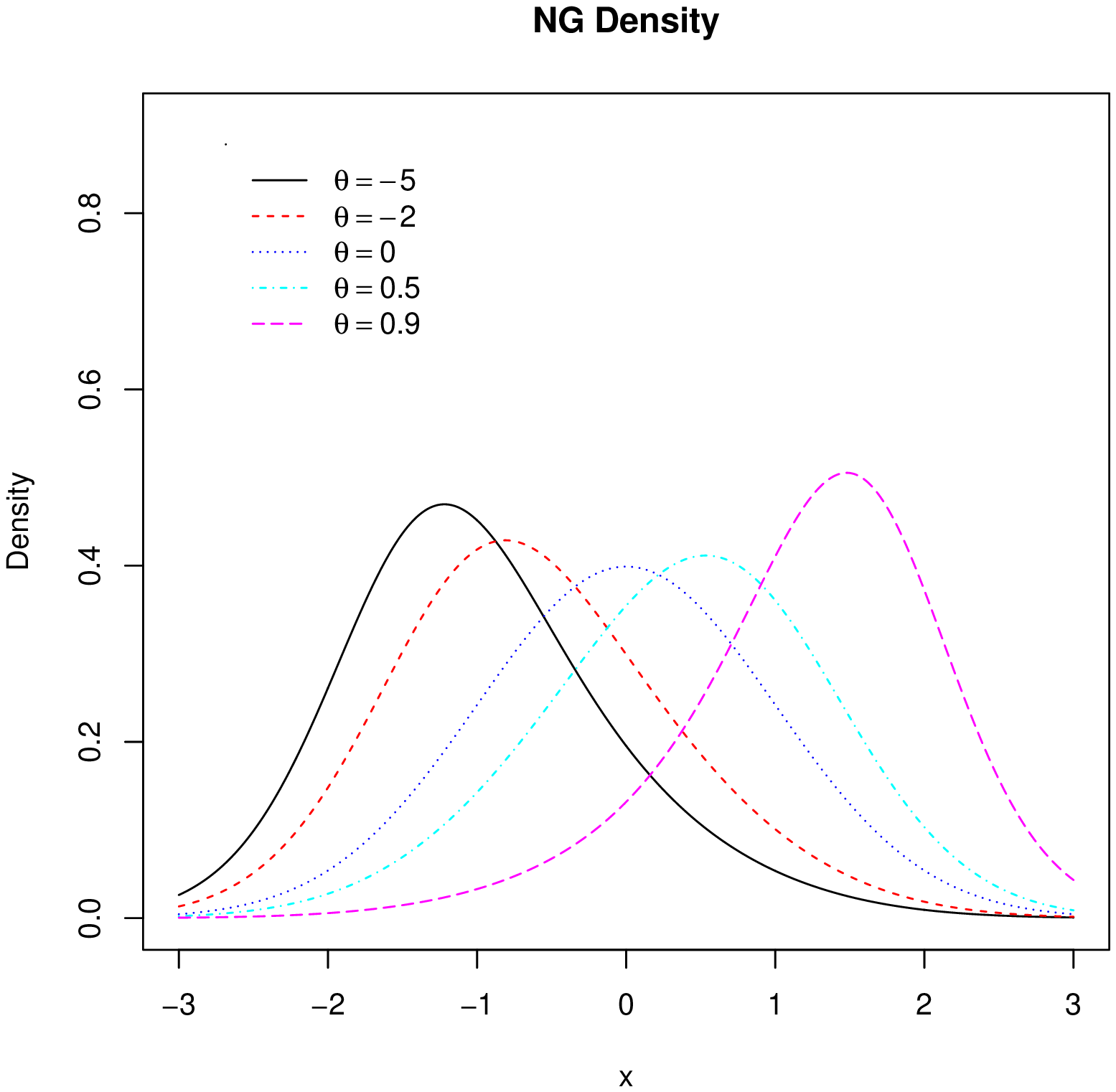}
\includegraphics[scale=0.40]{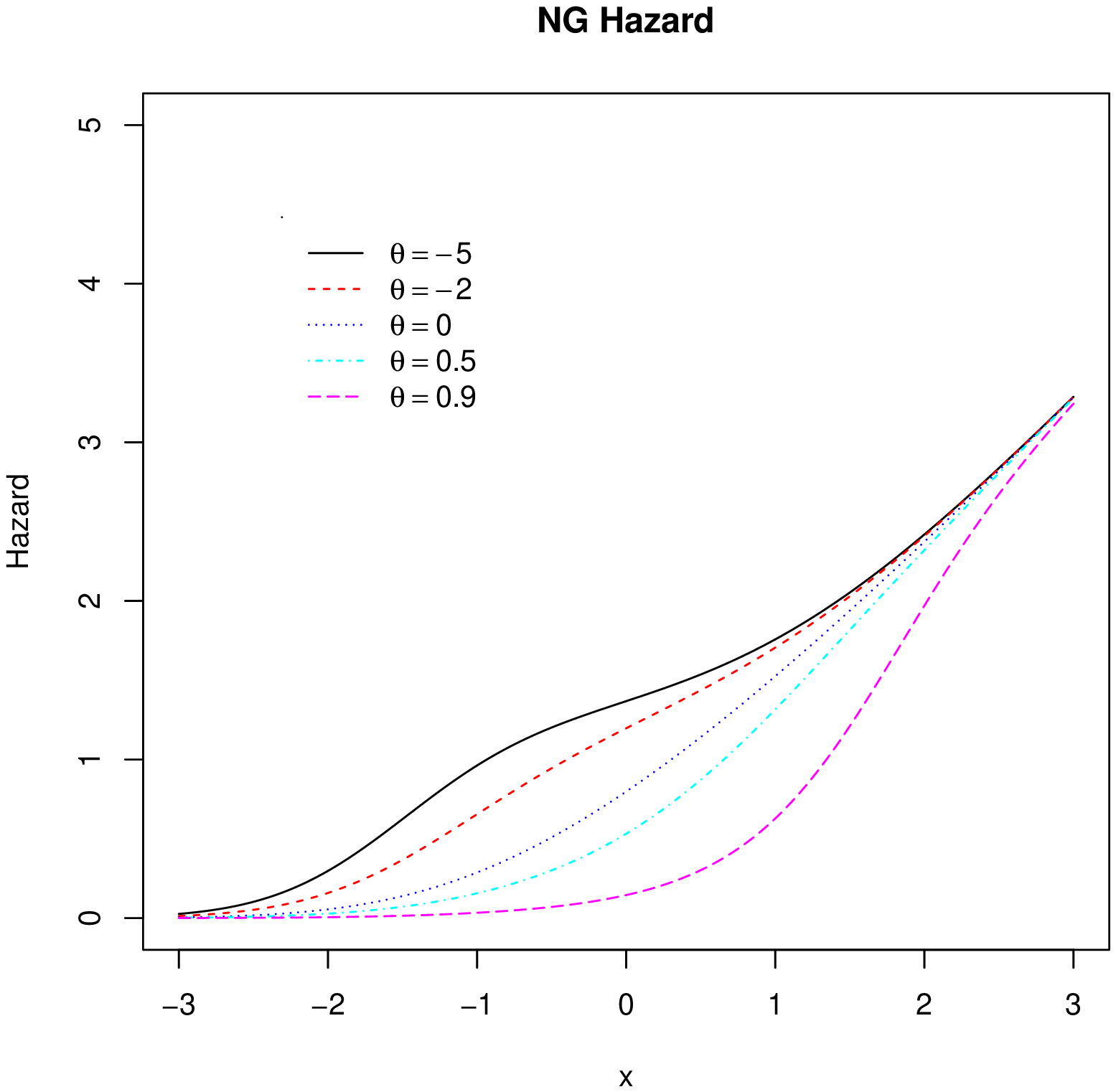}
\caption[]{Plots of density and hazard rate functions of NG distribution for selected
parameter values $\theta<1$ with $\mu=0,~\sigma=1$.}\label{phNG}
\end{figure}

\begin{theorem} \label{like}
If $Y_{1}\sim NG(0,1,\theta _{1})$ and $Y_{2}\sim NG(0,1,\theta
_{2})$, and $\theta _{1}>\theta _{2}$, then $Y_{2}<_{LR}Y_{1}$.
\end{theorem}

\begin{proof}
 The logarithm of the likelihood ratio
\begin{equation*}
\upsilon (y)=\log \frac{f(y;0,1,\theta _{1})}{f(y;0,1,\theta _{2})}=\log \frac{1-\theta _{1}}{1-\theta _{2}}
+2\log (1-\theta _{2}\Phi (y))-2\log (1-\theta _{1}\Phi (y)),
\end{equation*}%
is an increasing function of $y$ if $\theta _{1}>\theta _{2}$, since 
\begin{equation*}
\upsilon'(y)=\frac{2(\theta _{1}-\theta _{2})\phi (y)}{%
(1-\theta _{2}\Phi (y))(1-\theta _{1}\Phi (y))}>0,
\end{equation*}%
for all $y$. Therefore, the NG has the likelihood ratio ordering, which
implies it has the failure rate ordering as well as the stochastic ordering
and the mean residual life ordering.
\end{proof}
\begin{prop}
The moment generating function, mean and second central moment of $NG$ are given by 
\begin{equation*}
M_{Y}(t)=\exp \left( \frac{1}{2}t^{2}\right) \sum_{n=1}^{\infty } n(1-\theta )\theta
^{n-1}\Phi _{n-1}(1_{n-1}t;I_{n-1}+1_{n-1}1_{n-1}^{T}),
\end{equation*}
\begin{equation*}
E(Y)=\frac{1 }{2\sqrt{\pi }}\sum_{n=1}^{\infty }n(n-1)(1-\theta )\theta ^{n-1}\Phi
_{n-2}\left( \mathbf{0;}I_{n-2}+\frac{1}{2}1_{n-2}1_{n-2}^{T}\right),
\end{equation*}
and%
\begin{equation*}
E(Y^{2}) =1+\frac{1 }{4\sqrt{3}\pi}%
\sum_{n=1}^{\infty }n(n-1)(n-3)(1-\theta )\theta ^{n-1} \Phi
_{n-3}\left( \mathbf{0;}I_{n-3}+\frac{1}{3}1_{n-3}1_{n-3}^{T}\right),
\end{equation*}
respectively.
\end{prop}
Table \ref{Momtab} gives the first four moments, variance, skewness and kurtusis of the $NG(0,1,\theta)$ for different values $\theta<1$. One can see from this table that.\\
\begin{table}
\centering
\caption{The first four moments, variance, skewness and kurtusis of NG distribution for $\mu=0,~\sigma=1$ and different values $\theta$}

\begin{tabular}{|l|llllllll|}
\hline\\
&$\theta=-5$&$\theta=-2$&$\theta=-0.5$&$\theta=0$&$\theta=0.3$&$\theta=0.5$&$\theta=0.8$&$\theta=0.9$\\
\hline\hline\\
 $E(X)$& -0.9841& -0.6134&  -0.2284&    0 & 0.2010 & 0.3894 & 0.8884&  1.2445\\
$E(X^2)$ & 1.8465 & 1.3270   &1.0452  &  1 & 1.0350 & 1.1315 & 1.6887 & 2.3609\\
$E(X^3)$ &-3.1487 &-1.6981& -0.5795&    0 & 0.5083 & 1.0155 & 2.7254&  4.5206\\
$E(X^4)$  &7.2110  &4.4974   & 3.1974 &   3 & 3.1526 & 3.5829 & 6.3424& 10.313\\
$Var$  &0.8781&  0.9508&    0.9930&    1  &0.9946 & 0.9799&  0.8995 & 0.8123\\
$Sk$ & 0.4821&  0.3046 &   0.1141 &   0& -0.1004& -0.1942& -0.4371& -0.5999\\
$Kur$ &3.5440&  3.2104&   3.0291&    3&  3.0225&  3.0846&  3.4429 & 3.8702\\
\hline\\
\end{tabular}\label{Momtab}
\end{table}

Figure \ref{sNG} shows the skewness and kurtusis plot of the $NG(0,1,\theta)$ for different values $\theta<1$ with $\mu=0,~\sigma=1$.
\begin{figure}
\centering
\includegraphics[scale=0.40]{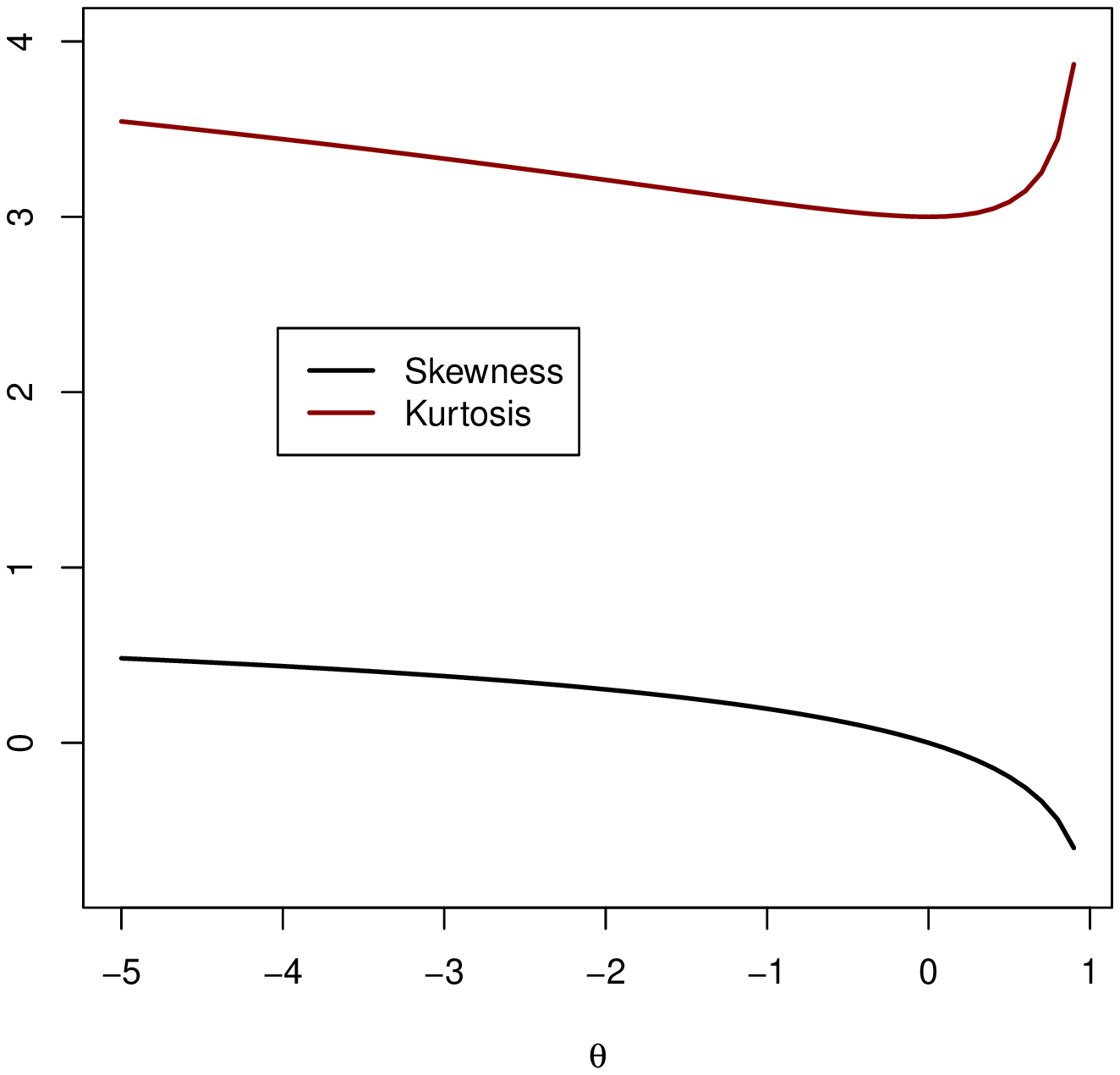}
\caption[]{Plots of skewness and kurtusis of NG distribution for selected
parameter values $\theta<1$.}\label{sNG}
\end{figure}

\subsection{Normal Poissen distribution}

The normal Poisson distribution is obtained when $a_{n}=\frac{1}{n!}$ and $C(\theta )=e^{\theta }-1$.
The cdf, pdf and hazard rate function of NP distribution are given by
\begin{equation*}
F(y;\mu,\sigma,\theta)=\frac{e^{\theta \Phi (\frac{y-\mu }{\sigma })}-1}{e^\theta-1},  y\mathbf{\in }\mathbb{R},\ \ \qquad \ 
\end{equation*}
\begin{equation}\label{NPpdf}
f(y;\mu,\sigma,\theta)=\frac{\theta }{\sigma }\frac{\phi (\frac{y-\mu }{\sigma })e^{\theta
\Phi (\frac{y-\mu }{\sigma })}}{e^{\theta }-1},
\end{equation}
and
\begin{equation*}
h(y;\mu,\sigma,\theta)=\frac{\theta }{\sigma }\frac{\phi (\frac{y-\mu }{\sigma })e^{\theta
\Phi (\frac{y-\mu }{\sigma })})}{e^{\theta }-e^{\theta \Phi (\frac{y-\mu }{%
\sigma })})},
\end{equation*}
respectively, where $y \in \mathbb{R}$ and $\theta \in (0,\infty)$. We use the notation $Y \sim NP(\mu,\sigma,\theta)$ when the random variable $Y$ has NP distribution with location $\mu$, scale $\sigma$ and shape parameter $\theta$. \\
\begin{remark}
Even when $\theta<0$, Equation \eqref{NPpdf} is also a density function.
We can then define the NG distribution by Equation \eqref{NPpdf}
for any $\theta \in\mathbb{R}$- $\left\{ 0\right\} $ .
\end{remark}
Figure \ref{phNP} shows the NP density and hazard rate functions for selected values $\theta$ where $\mu=0$ and $\sigma=1$.
\begin{figure}
\centering
\includegraphics[scale=0.40]{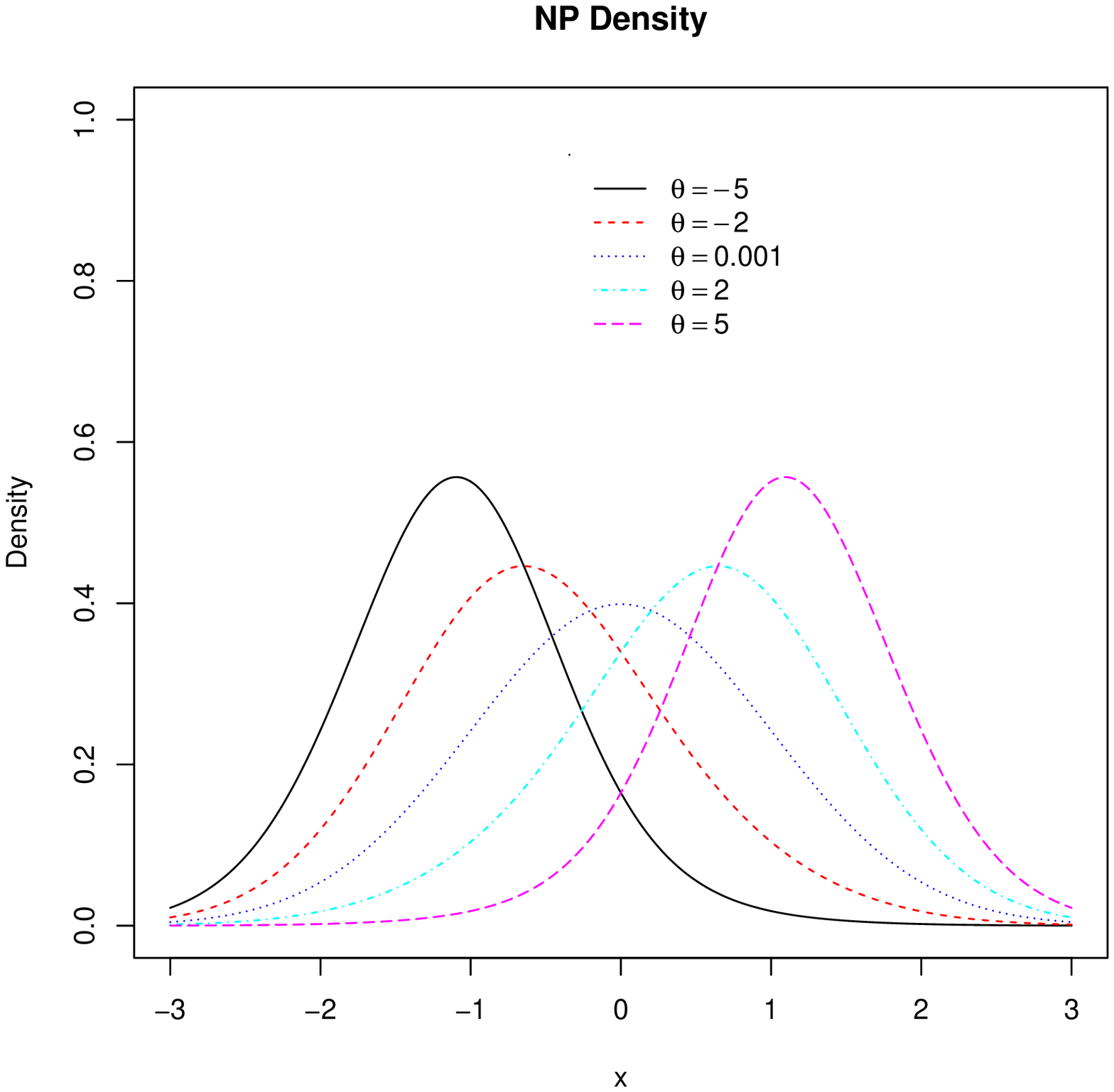}
\includegraphics[scale=0.40]{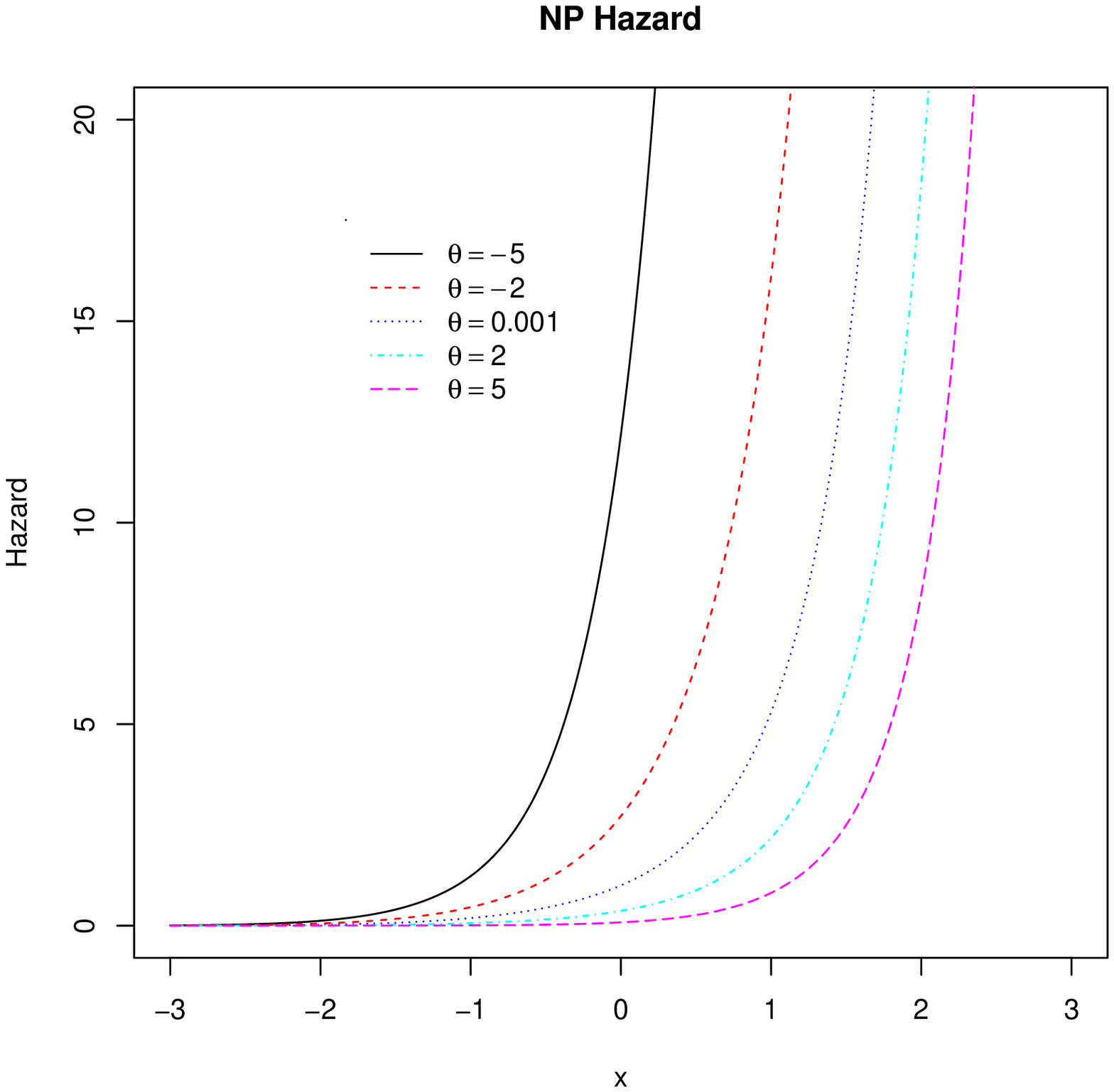}
\caption[]{Plots of density and hazard rate functions of NP distribution for selected
parameter values $\theta<1$ with $\mu=0,~\sigma=1$.}\label{phNG}
\end{figure}

\begin{theorem}
If $Y_{1}\sim NP(\mu,\sigma,\theta _{1})$ and $Y_{2}\sim NP(\mu,\sigma,\theta
_{2})$, and $\theta _{1}>\theta _{2}$, then $Y_{2}<_{LR}Y_{1}$.
\end{theorem}
\begin{proof}
The proof is similar to the proof of Theorem \ref{like} and is omitted.
\end{proof}

\begin{prop}
The moment generating function, mean and second central moment of NP are given by
\begin{equation*}
M_{Y}(t)=\exp \left( \frac{1}{2}t^{2}\right) \sum_{n=1}^{\infty }\frac{%
\theta ^{n}}{n!(e^{\theta }-1)}\times n\Phi
_{n-1}(1_{n-1}t;I_{n-1}+1_{n-1}1_{n-1}^{T}),
\end{equation*}
\begin{equation*}
E(Y)=\frac{1 }{2\sqrt{\pi }}\sum_{n=1}^{\infty }\frac{%
n(n-1)\theta ^{n}}{n!(e^{\theta }-1)}\Phi _{n-2}\left( \mathbf{0;}%
I_{n-2}+\frac{1}{2}1_{n-2}1_{n-2}^{T}\right) ,
\end{equation*}
and
\begin{eqnarray*}
E(Y^{2})=1+\frac{1}{4\sqrt{3}\pi }\sum_{n=1}^{\infty }\frac{\theta ^{n}}{%
n!(e^{\theta }-1)}\times n(n-1)(n-3)\Phi _{n-3}\left( \mathbf{0};I_{n-3}+%
\frac{1}{3}1_{n-3}1_{n-3}^{T}\right) .
\end{eqnarray*}
\end{prop}
Table \ref{MomtabNP} gives the first four moments, variance, skewness and kurtusis of the $NP(0,1,\theta)$ for different values $\theta>0$. One can see from this table that.\\

\begin{table}
\centering
\caption{The first four moments, variance, skewness and kurtusis of NP distribution for $\mu=0,~\sigma=1$ and different values $\theta$}
\begin{tabular}{|l|llllllll|}
\hline\\
&$\theta=0.01$&$\theta=0.3$&$\theta=0.5$&$\theta=0.8$&$\theta=1$&$\theta=3$&$\theta=6$&$\theta=10$\\
\hline\hline\\
 $E(X)$& 0.0028&  0.0845 & 0.1405& 0.2236 & 0.2781 & 0.7541 & 1.1997 & 1.5045\\
$E(X^2)$ & 1.0000 & 1.0041 & 1.0114 & 1.0290 & 1.0450 & 1.3477  &1.9673  &2.6533\\
$E(X^3)$ &0.0071  &0.2114 & 0.3520 & 0.5617  &0.7003 & 2.0013  &3.5904 & 5.2127\\
$E(X^4)$  &3.0000  &3.0179 & 3.0495  &3.1259&  3.1954&  4.5372&  7.4821& 11.2262\\
$Var$  &1.0000 & 0.9970 & 0.9917&  0.9790&  0.9677 & 0.7790&  0.5279 & 0.3898\\
$Sk$ & -0.0014 &-0.0421& -0.0697 &-0.1097 &-0.1349 &-0.2764& -0.0956 & 0.1973\\
$Kur$ &3.0000 & 3.0074 & 3.0204 & 3.0515&  3.0792 & 3.5076 & 3.6846 & 3.4236\\
\hline\\
\end{tabular}\label{MomtabNP}
\end{table}

Figure \ref{sNP} shows the skewness and kurtusis plot of the $NP(0,1,\theta)$ for different values $\theta$.
\begin{figure}
\centering
\includegraphics[scale=0.40]{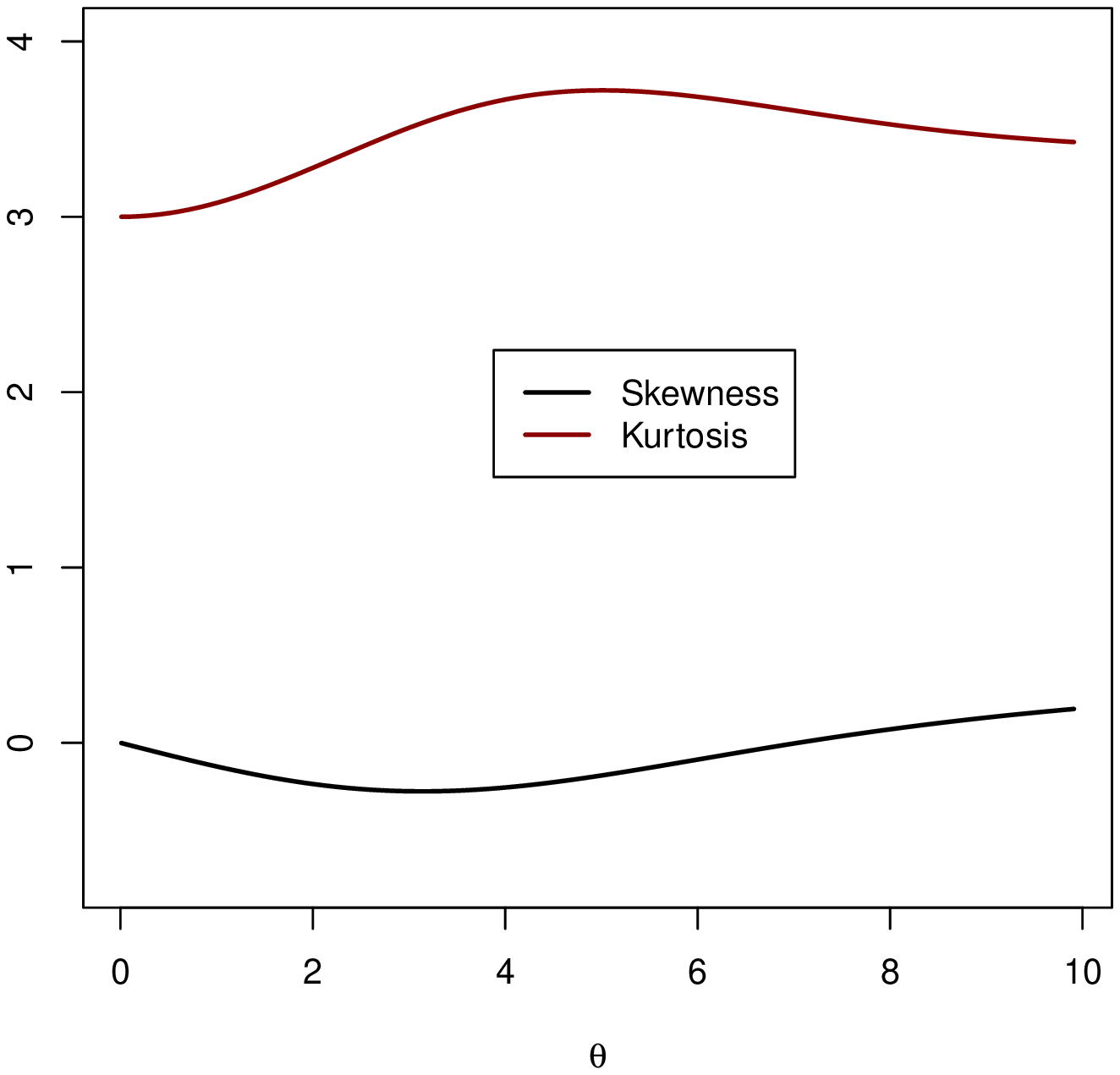}
\caption[]{Plots of skewness and kurtusis of NP distribution for selected
parameter values $\theta$.}\label{sNP}
\end{figure}

\subsection{Normal-binomial distributions}

\bigskip When $a_{n}=\left( 
\begin{array}{c}
m \\ 
n%
\end{array}%
\right) $ and $C(\theta )=(\theta +1)^{m}-1(\theta >0),$ where $m$ $(n\leq
m) $ is the number of replicas. we obtain the Normal-binomial distribution ($%
NB$) with cdf%
\begin{equation}
F(y;\mu,\sigma,\theta)=\frac{(\theta \Phi (\frac{y-\mu }{\sigma })+1)^{m}-1}{(\theta +1)^{m}-1}%
,\;y\mathbf{\in }\mathbb{R},\ \ \qquad \ 
\end{equation}%
the pdf and hazard rate function are%
\begin{equation}\label{NBpdf}
f(y;\mu,\sigma,\theta)=\frac{\theta }{\sigma }\phi \left(\frac{y-\mu }{\sigma }\right)\frac{m(\theta \Phi
(\frac{y-\mu }{\sigma })+1)^{m-1}}{(\theta +1)^{m}-1},
\end{equation}%
and%
\begin{equation*}
h(y;\mu,\sigma,\theta)=\frac{\theta }{\sigma }\phi \left( \frac{x-\mu }{\sigma }\right) \frac{%
m(\theta \Phi \left( \frac{y-\mu }{\sigma }\right) +1)^{m-1}}{(\theta
+1)^{m}-(\theta \Phi \left( \frac{y-\mu }{\sigma }\right) +1)^{m}},\ \ \ \ 
\end{equation*}%
respectively, where $y \in \mathbb{R}$ and $\theta \in (0,\infty)$. We use the notation $Y \sim NB(\mu,\sigma,\theta)$ when the random variable $Y$ has NB distribution with location $\mu$, scale $\sigma$ and shape parameter $\theta$. \\
\begin{remark}
Even when $\theta<0$, Equation \eqref{NBpdf} is also a density function.
We can then define the NB distribution by Equation \eqref{NBpdf}
for any $\theta \in \mathbb{R} $-$\{0\}$.
\end{remark}
Figure \ref{phNB} shows the NB density and hazard rate functions for selected values $\theta$ where $\mu=0$ and $\sigma=1$.

\begin{figure}
\centering
\includegraphics[scale=0.40]{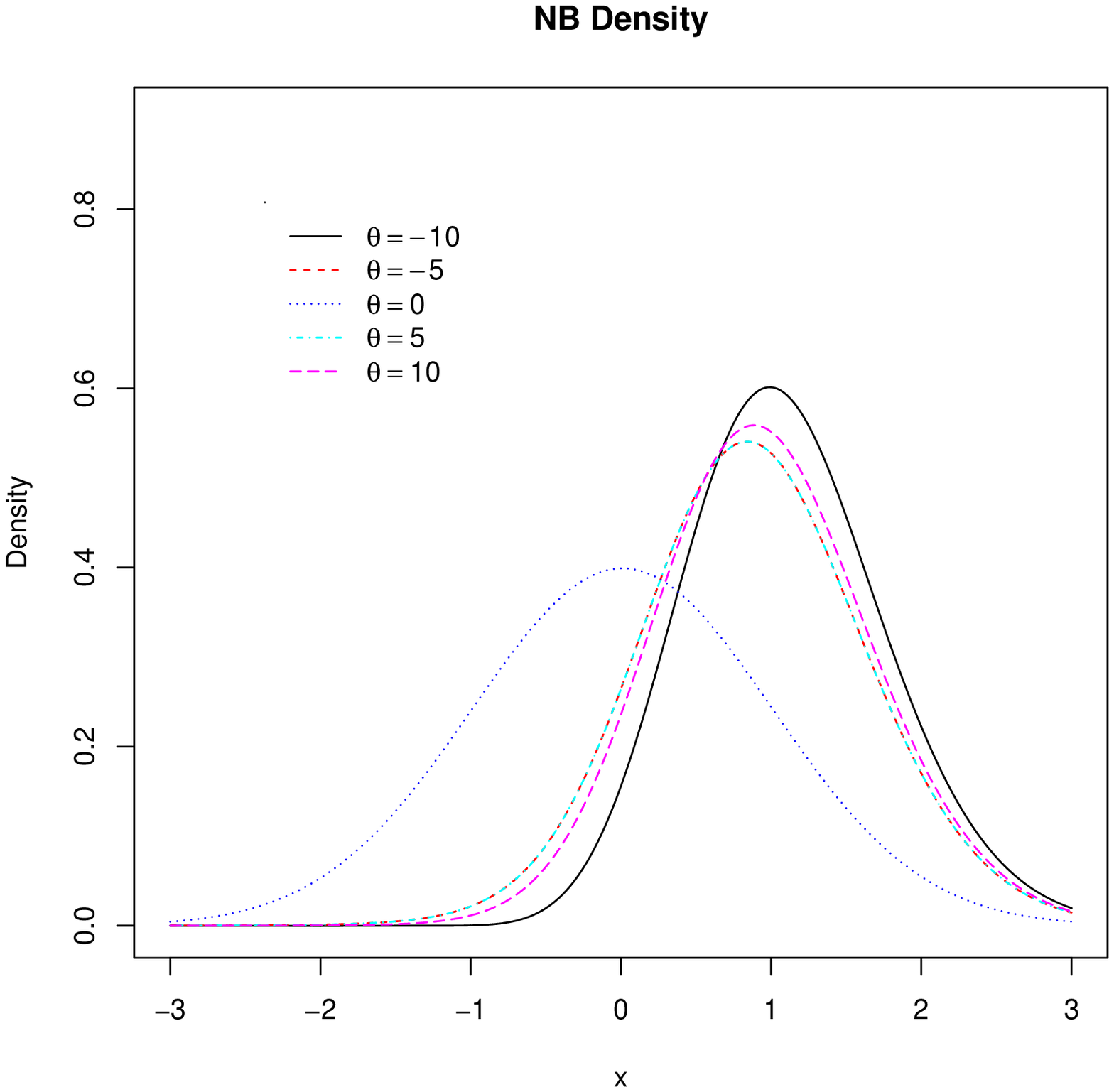}
\includegraphics[scale=0.40]{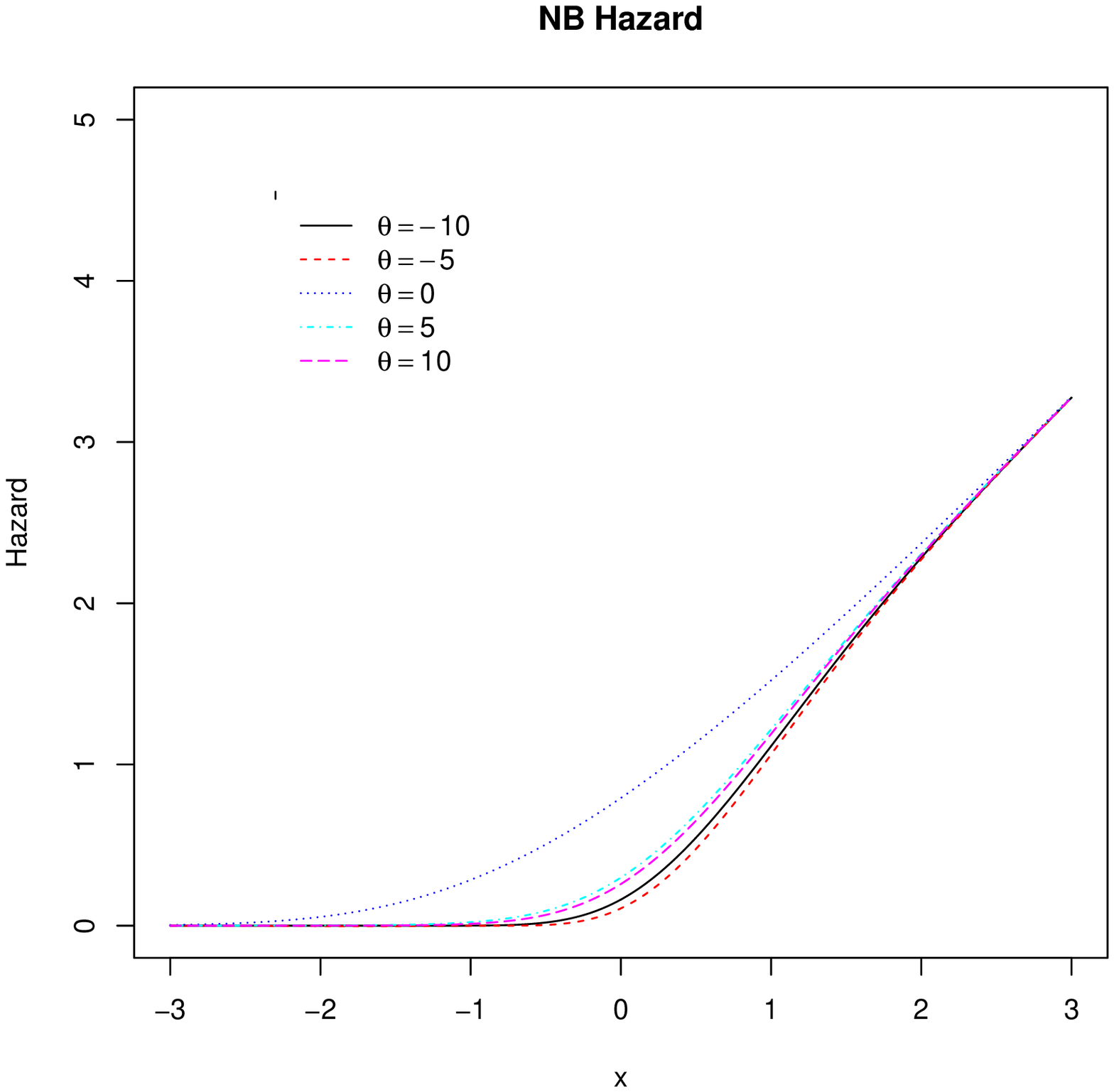}
\caption[]{Plots of density and hazard rate functions of NB distribution for selected
parameter values $\theta$ with $\mu=0,~\sigma=1$.}\label{phNB}
\end{figure}
\begin{prop}
The moment generating function, mean and second central moment of NB are given by
\begin{equation*}
M_{X}(t)=\exp \left( \frac{1}{2}t^{2}\right) \sum_{n=1}^{\infty }\frac{%
\theta ^{n}}{(\theta +1)^{m}-1}\times n\Phi
_{n-1}(1_{n-1}t;I_{n-1}+1_{n-1}1_{n-1}^{T}),
\end{equation*}%

\begin{equation*}
E(Y)=\frac{1 }{2\sqrt{\pi }}\sum_{n=1}^{\infty }\left( 
\begin{array}{c}
m \\ 
n%
\end{array}%
\right) \frac{\theta ^{n}}{(\theta +1)^{m}-1}n(n-1)\Phi
_{n-2}\left( \mathbf{0;}I_{n-2}+\frac{1}{2}1_{n-2}1_{n-2}^{T}\right) ,
\end{equation*}%
and
\begin{eqnarray*}
E(Y^{2})=1+\frac{1}{4\sqrt{3}\pi }\sum_{n=1}^{\infty }\left( 
\begin{array}{c}
m \\ 
n%
\end{array}%
\right) \frac{n(n-2)(n-3)\theta ^{n}}{(\theta +1)^{m}-1}\Phi _{n-3}\left( 
\mathbf{0;}I_{n-3}+\frac{1}{3}1_{n-3}1_{n-3}^{T}\right),
\end{eqnarray*}
\end{prop}
\bigskip

\subsection{Normal-logarithmic distributions}

When $a_{n}=\frac{1}{n}$ and $c(\theta )=-\log (1-\theta )(0<\theta <1)$ we
obtain the Normal-logarithmic distribution ($NL$) with cdf%
\begin{equation*}
F(y;\mu,\sigma,\theta)=\frac{\log (1-\theta \Phi \left( \frac{y-\mu }{\sigma }\right) )}{\log
(1-\theta )},\;y\mathbf{\in }\mathbb{R}.
\end{equation*}%
the pdf and hazard rate function are%
\begin{equation}\label{NLpdf}
f(y;\mu,\sigma,\theta)=\frac{\frac{\theta }{\sigma }\phi \left( \frac{y-\mu }{\sigma }\right) 
}{(\theta \Phi \left( \frac{y-\mu }{\sigma }\right) -1)\log (1-\theta )},
\end{equation}%
and%
\begin{equation*}
h(y;\mu,\sigma,\theta)=\frac{\frac{\theta }{\sigma }\phi \left( \frac{y-\mu }{\sigma }\right) 
}{(\theta \Phi \left( \frac{y-\mu }{\sigma }\right) -1)\log \frac{1-\theta }{%
1-\theta \Phi \left( \frac{y-\mu }{\sigma }\right) }},
\end{equation*}%
respectively, where $y \in \mathbb{R}$ and $\theta \in (0,1)$. We use the notation $Y \sim NL(\mu,\sigma,\theta)$ when the random variable $Y$ has NL distribution with location $\mu$, scale $\sigma$ and shape parameter $\theta$. \\
\begin{remark}
Even when $\theta<0$, Equation \eqref{NLpdf} is also a density function.
We can then define the NL distribution by Equation \eqref{NLpdf}
for any $\theta \in (-\infty ,0)\cup (0,1)$.
\end{remark}
Figure \ref{phNL} shows the NL density and hazard rate functions for selected values $\theta$ where $\mu=0$ and $\sigma=1$.
\begin{figure}
\centering
\includegraphics[scale=0.40]{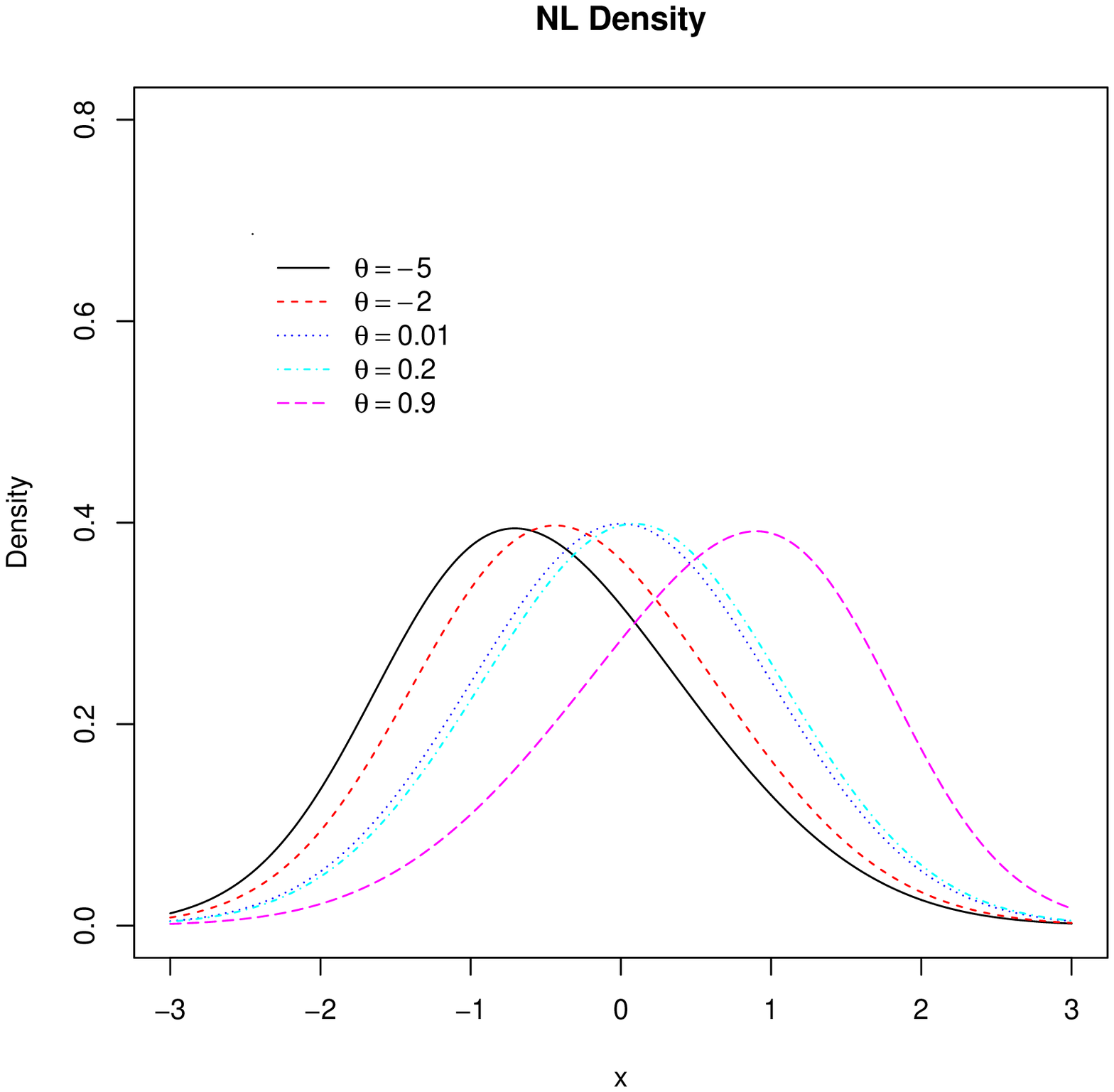}
\includegraphics[scale=0.40]{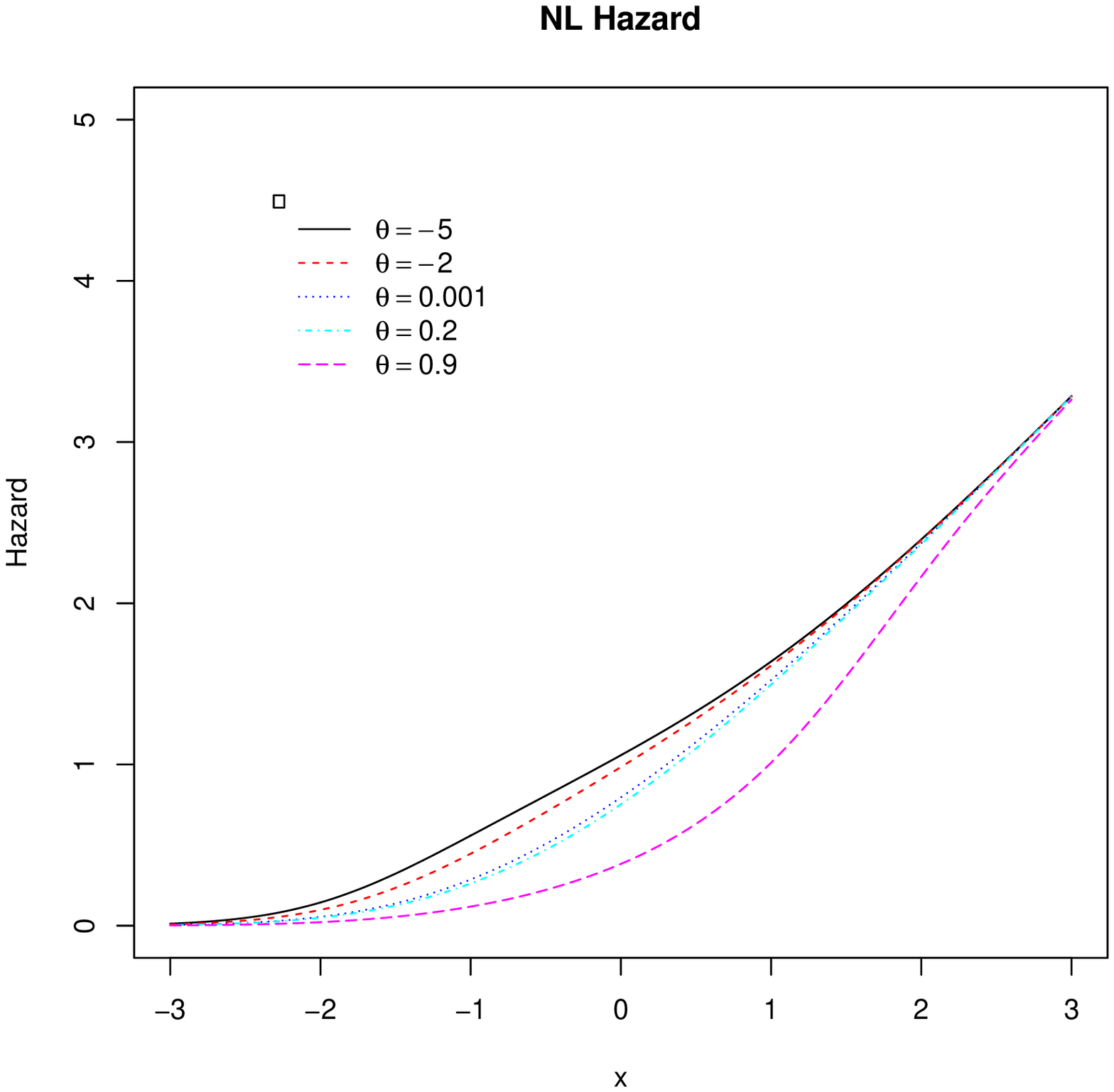}
\caption[]{Plots of density and hazard rate functions of NL distribution for selected
parameter values $\theta$ with $\mu=0,~\sigma=1$.}\label{phNL}
\end{figure}
\bigskip

\begin{prop}
The moment generating function, mean and second central moment of NL are given by
\begin{equation*}
M_{Y}(t)=\sum_{n=1}^{\infty }\frac{a_{n}\theta ^{n}}{C(\theta )}%
M_{Y(n)}(t)=\exp (\frac{1}{2}t^{2})\sum_{n=1}^{\infty }\frac{\theta ^{n}}{%
\log (1-\theta )}\times \Phi _{n-1}(1_{n-1}t;I_{n-1}+1_{n-1}1_{n-1}^{T}),
\end{equation*}%

\begin{equation*}
E(Y)=-\frac{1 }{2\sqrt{\pi }}\sum_{n=1}^{\infty }\frac{\theta ^{n}}{n\log (1-\theta )}(n-1)\Phi _{n-2}\left( \mathbf{0;}I_{n-2}+\frac{1}{2}1_{n-2}1_{n-2}^{T}%
\right),
\end{equation*}%
and
\begin{eqnarray*}
E(Y^{2})=1-\frac{1}{4\sqrt{3}\pi }\sum_{n=1}^{\infty }\frac{%
n(n-2)(n-3)\theta ^{n}}{n\log (1-\theta )}\Phi _{n-3}\left( \mathbf{0;}%
I_{n-3}+\frac{1}{3}1_{n-3}1_{n-3}^{T}\right) ,
\end{eqnarray*}
respectively.
\end{prop}

\section{Properties of sub model of NPS distristribution}
In this section we present some additional and useful properties of the sub models of NPS distribution.  
\begin{prop}
For NG, NB and NL distributions we have 
\begin{equation*}
F(y;0,1,\theta )=1-F\left( -y;0,1,\frac{\theta }{\theta -1}\right),
\end{equation*}%
and for NP distribution we have 
\begin{equation*}
F(y;0,1,\theta )=1-F\left( -y;0,1,-\theta \right) ,
\end{equation*}%
\end{prop}
\begin{proof}
We shall prove for NG distribution. Proofs of other distributions are similar. From \ref{NGcdf}, it can be found that
\begin{eqnarray*}
1-F(-y;0,1,\frac{\theta }{\theta -1}) &=&1+\frac{\frac{1}{\theta -1}%
\Phi (-y)}{1-\frac{\theta }{\theta -1}\Phi (-y)}=1+\frac{1-\Phi (y)}{\theta
\Phi (y)-1} \\
&=&\frac{(1-\theta )\Phi (y)}{1-\theta \Phi (y)}=F(y;0,1,\theta )
\end{eqnarray*}
and hence the proof is compeleted. 
\end{proof}

In the following proposition we give approximations for first and second moments around the origin of NG, NP and NB distributions.  
\begin{prop} We have\\
(i) If $Y\sim NG(\mu ,\sigma ,\theta )$, then $E(Y)$ and $E(Y^{2})$ are approximated as:\\ 
\begin{eqnarray*}
E(Y)\simeq \frac{1}{2\theta ^{2}}\left\{ 2\theta ^{2}\mu +\sqrt{2\pi }%
\sigma \left( 2\log \left( 1-\theta \right) \left( 1-\theta \right) +2\theta
-\theta ^{2}\right) \right\}, 
\end{eqnarray*}
\begin{eqnarray*}
E(Y^{2})&\simeq & \frac{1}{4\theta \left( \theta -1\right) ^{2}}\Big\{ \theta
^{3}\left( \pi \sigma ^{2}+2\mu ^{2}-2\sqrt{2}\sqrt{\pi }\sigma \mu \right)
+16\pi \theta \sigma ^{2} +28\sigma \theta ^{2}\left( \sqrt{2\pi }\mu -2\pi \sigma \right)\ \medskip\\
&+& 8\sigma \log \left( 1-\theta \right) \left( \sqrt{2\pi }\theta \mu \left(
1-\theta \right) +\left( \theta \left( \theta -3\right) +2\right) \pi \sigma
\right) \Big\},
\end{eqnarray*}
(ii) If $Y\sim NP(\mu,\sigma ,\theta )$, then $E(Y)$ and $E(Y^{2})$ are approximated as:
\begin{eqnarray*}
E(Y)\simeq \allowbreak \frac{1}{2\theta \sigma \left( e^{\theta
}-1\right) }\left\{ \left( 2+\theta \right) \sigma \sqrt{2\pi }-2\theta \mu +%
\left( 2\theta \mu +\left( \theta -2\right) \sigma \sqrt{2\pi }\right)
e^{\theta }\right\} ,
\end{eqnarray*}
\begin{eqnarray*}
E(Y^{2})&\simeq & \frac{1}{2\theta ^{2}\left( e^{\theta }-1\right) }\Big\{
\sigma \mu \theta \left( 4\sqrt{2\pi }+2\sqrt{2\pi }\theta \right) -8\pi
\sigma ^{2}-\theta ^{2}\left( \pi \sigma ^{2}+2\mu ^{2}\right) -4\pi \theta
\sigma ^{2}\ \\
&+&\left(  2\theta ^{2}\mu ^{2}+8\pi \sigma ^{2}-4\pi \theta \sigma
^{2}+\pi \theta ^{2}\sigma ^{2}+2\sqrt{2\pi }\theta ^{2}\sigma \mu -4\sqrt{%
2\pi }\theta \sigma \mu \right) e^{\theta } \Big\}.
\end{eqnarray*}
(iii) If $Y\sim NB(\mu,\sigma ,\theta )$, then $E(Y)$ and $E(Y^{2})$ are approximated as:

\begin{eqnarray*}
E(Y) &\simeq &\frac{1}{2\theta \left( \left( \theta +1\right) ^{m}-1\right) \left(
m+1\right) }\Big\{ 2\sqrt{2\pi }\sigma +\left( \sqrt{2\pi }\sigma
(1+m)-2\mu (1+m)\right) \theta\\
 &-&\left[ 2\sqrt{2\pi }\sigma \left( 1+m\right)
+\left( \sqrt{2\pi }\sigma \left( 1+m\right) -2\mu \left( 1+m\right) \right)
\theta -2\sqrt{2\pi }m\sigma \left( \theta +1\right) \right] \left( \theta
+1\right) ^{m}\Big\} .
\end{eqnarray*}
\begin{eqnarray*}
E(Y^{2}) &\simeq &\frac{1}{\theta ^{2}\left( \theta +1\right) ^{m}-1
(m+2) ( m+1) }\Big\{ ( -2\mu ^{2}+\sqrt{2\pi }
m^{2}\sigma \mu -\frac{3}{2}\pi m\sigma ^{2}-\pi \sigma ^{2}-3m\mu^{2}\\
&+& 3\sqrt{2\pi }m\sigma \mu -m^{2}\mu ^{2}-\frac{1}{2}\pi m^{2}\sigma
^{2}+2\sqrt{2\pi }\sigma \mu ) \theta ^{2}+( 4\sqrt{2\pi }\sigma
\mu +2\sqrt{2\pi }m\sigma \mu -2\pi m\sigma ^{2}-4\pi \sigma ^{2})
\theta  \\
&+&[( 2\mu ^{2}+3m\mu ^{2}+\pi \sigma ^{2}+\sqrt{2\pi }m\sigma \mu
+\sqrt{2\pi }m^{2}\sigma \mu -2\sqrt{2\pi }\sigma \mu -\frac{1}{2}\pi
m\sigma ^{2}+\frac{1}{2}\pi m^{2}\sigma ^{2}+m^{2}\mu ^{2}) \theta
^{2} \\
&+&4\pi \sigma ^{2}+( 4\pi \sigma ^{2}-2\pi m\sigma ^{2}-4
\sqrt{2\pi }\sigma \mu -2\sqrt{2\pi }m\sigma \mu ) \theta]
\left( \theta +1\right) ^{m}\Big\}.
\end{eqnarray*}

\end{prop}
\begin{proof} For (i), note first that
\[
E(Y)=\int_{-\infty }^{\infty }\frac{(1-\theta )y\phi (\mu ,\sigma
,\theta )}{(1-\theta \Phi (\mu ,\sigma ,\theta ))^{2}}. 
\]%
After the change of variable $u=1-\theta \Phi (\mu ,\sigma ,\theta )$, we
obtain $y=\Phi ^{-1}\left( \frac{1-u}{\theta };\mu ,\sigma \right) =\mu +
\sqrt{2}\sigma {erf}^{-1}\left( 2\left( \frac{1-u}{\theta }\right)
-1\right) $. Now because ${erf}^{-1}(z)=z\frac{\sqrt{\pi }}{2}+O(z^{3})$%
, we can write ${erf}^{-1}(z)\simeq z\frac{\sqrt{\pi }}{2}$. Therefore
we have%
\[
\mu _{1}=-\frac{1-\theta }{\theta }\int_{1}^{1-\theta }\frac{\mu +\sqrt{%
\frac{\pi }{2}}\sigma \left( 2\left( \frac{1-u}{\theta }\right) -1\right) }{%
u^{2}}du. 
\]%
the result is obtained by solving the integral. Finally, $E(Y^{2})$ is derived in the same manner after simple computation. Parts (ii) and (iii) follow in a same way.
\end{proof}

\section{Estimation and inference}
In this section, we discuss the estimation of the parameters of $%
NPS $ distribution. let $Y_{1},Y_{2},...,Y_{n}$ be a random sample with
observed values $y_{1},y_{2},...,y_{n}$ from a $NPS(\mu ,\sigma ,\theta )$
and $\Psi =(\mu ,\sigma ,\theta )^{T}$ be a parameter vector. The total
log-likelihood function is given by
\begin{equation*}
l_{n}=l_{n}(\Theta ;\mathbf{y})=n\log (\theta )-n\log (\sigma )-\frac{n}{2}%
\log (2\pi )-\frac{1}{2}\sum_{i=1}^{n}z_{i}^{2}+\sum_{i=1}^{n}\log
(C^{^{\prime }}(\theta \Phi (z_{i}))-n\log (C(\theta )),
\end{equation*}%
where $z_{i}=\frac{y_{i}-\mu }{\sigma }$. The maximum likelihood estimation
(MLE) of $\Psi$, say $\widehat{\Psi }$, is obtained by solving the nonlinear
system $\left( \frac{\partial l_{n}}{\mu },\frac{\partial l_{n}}{\sigma },%
\frac{\partial l_{n}}{\theta }\right) ^{T}=0$, where
\begin{eqnarray*}
\frac{\partial l_{n}}{\partial \mu } &=&\frac{1}{\sigma }\sum_{i=1}^{n}z_{i}-%
\frac{\theta }{\sigma }\sum_{i=1}^{n}\frac{\phi (z_{i})C^{^{\prime \prime
}}(\theta \Phi (z_{i}))}{C^{^{\prime }}(\theta \Phi (z_{i}))}, \\
\frac{\partial l_{n}}{\partial \sigma } &=&-\frac{n}{\sigma }+\frac{1}{%
\sigma }\sum_{i=1}^{n}z_{i}^{2}-\frac{\theta }{\sigma }\sum_{i=1}^{n}\frac{%
\allowbreak z_{i}\phi (z_{i})C^{^{\prime \prime }}(\theta \Phi (z_{i}))}{%
C^{^{\prime }}(\theta \Phi (z_{i}))}, \\
\frac{\partial l_{n}}{\partial \theta } &=&\frac{n}{\theta }+\sum_{i=1}^{n}%
\frac{\Phi (z_{i})C^{^{\prime \prime }}(\theta \Phi (z_{i}))}{C^{^{\prime
}}(\theta \Phi (z_{i}))}-\frac{nC^{^{\prime }}(\theta )}{C(\theta )}.
\end{eqnarray*}%
The solution of this nonlinear system of equation has not a closed form. The
observed information matrix is obtained for approximate confidence intervals
and hypothesis tests of the vector. The $3\times 3$ observed information
matrix is%
\begin{equation*}
I_{n}\left( \Psi \right) =-\left[ 
\begin{array}{ccc}
I_{\mu \mu } & I_{\mu \sigma } & I_{\mu \theta } \\ 
I_{\mu \sigma } & I_{\sigma \sigma } & I_{\sigma \theta } \\ 
I_{\mu \theta } & I_{\sigma \theta } & I_{\theta \theta }%
\end{array}%
\right],
\end{equation*}%
where%
\begin{eqnarray*}
I_{\mu \mu } &=&-\frac{n}{\sigma ^{2}}-\frac{\theta }{\sigma ^{2}}%
\sum_{i=1}^{n}\frac{\left[ z_{i}\phi (z_{i})C^{^{\prime \prime }}(\theta
\Phi (z_{i}))-\theta C^{^{\prime \prime \prime }}(\theta \Phi (z_{i})\phi
^{2}(z_{i})\right] C^{^{\prime }}(\theta \Phi (z_{i}))+\theta \left(
C^{^{\prime \prime }}(\theta \Phi (z_{i}))\right) ^{2}\phi ^{2}(z_{i})}{%
\left( C^{^{\prime }}(\theta \Phi (z_{i}))\right) ^{2}}, \\
&& \\
I_{\mu \sigma } &=&-\frac{2}{\sigma ^{2}}\sum_{i=1}^{n}z_{i}+\frac{\theta }{%
\sigma ^{2}}\sum_{i=1}^{n}\frac{\phi (z_{i})C^{^{\prime \prime }}(\theta
\Phi (z_{i}))}{C^{^{\prime }}(\theta \Phi (z_{i}))} \\
&&-\frac{\theta }{\sigma ^{2}}\sum_{i=1}^{n}\frac{\left[ z_{i}^{2}\phi
(\allowbreak z_{i})C^{^{\prime \prime }}(\theta \Phi (\allowbreak
z_{i}))-\theta \allowbreak z_{i}\phi ^{2}(\allowbreak z_{i})C^{^{\prime
\prime \prime }}(\theta \Phi (\allowbreak z_{i}))\right] C^{^{\prime
}}(\theta \Phi (\allowbreak z_{i}))+\theta z_{i}\phi ^{2}(\allowbreak
z_{i})\left( C^{^{\prime \prime }}(\theta \Phi (\allowbreak z_{i})\right)
^{2}}{\left( C^{^{\prime }}(\theta \Phi (\allowbreak z_{i}))\right) ^{2}}, \\
I_{\mu \theta } &=&-\frac{1}{\sigma }\sum_{i=1}^{n}\frac{\phi
(z_{i})C^{^{\prime \prime }}(\theta \Phi (z_{i}))}{C^{^{\prime }}(\theta
\Phi (z_{i}))} \\
&&-\frac{\theta }{\sigma }\sum_{i=1}^{n}\frac{\Phi (z_{i})\phi
(z_{i})C^{^{\prime \prime \prime }}(\theta \Phi (z_{i}))C^{^{\prime
}}(\theta \Phi (z_{i}))-\Phi (z_{i})\phi (z_{i})\left( C^{^{\prime \prime
}}(\theta \Phi (z_{i})\right) ^{2}}{\left( C^{^{\prime }}(\theta \Phi
(z_{i})\right) ^{2}}, \\
&&
\end{eqnarray*}%
\begin{eqnarray*}
I_{\sigma \sigma } &=&\frac{n}{\sigma ^{2}}-\frac{3}{\sigma ^{2}}%
\sum_{i=1}^{n}z_{i}^{2}+\frac{\theta }{\sigma ^{2}}\sum_{i=1}^{n}\frac{%
\allowbreak z_{i}\phi (z_{i})C^{^{\prime \prime }}(\theta \Phi (z_{i}))}{%
C^{^{\prime }}(\theta \Phi (z_{i}))} \\
&&+\frac{\theta }{\sigma ^{2}}\sum_{i=1}^{n}\frac{\left[ \left( \allowbreak
z_{i}^{3}\phi (z_{i})-\allowbreak z_{i}\phi (z_{i})\right) C^{^{\prime
\prime }}(\theta \Phi (z_{i})-\theta \allowbreak z_{i}^{2}\phi
^{2}(z_{i})C^{^{\prime \prime \prime }}(\theta \Phi (z_{i})\right]
C^{^{\prime }}(\theta \Phi (z_{i}))}{\left( C^{^{\prime }}(\theta \Phi
(z_{i}))\right) ^{2}} \\
&&-\frac{\theta ^{2}}{\sigma ^{2}}\sum_{i=1}^{n}\frac{\allowbreak
z_{i}^{2}\phi ^{2}(z_{i})\left( C^{^{\prime \prime }}(\theta \Phi
(z_{i})\right) ^{2}}{\left( C^{^{\prime }}(\theta \Phi (z_{i}))\right) ^{2}},
\\
&& \\
I_{\sigma \theta } &=&-\frac{1}{\sigma ^{2}}\sum_{i=1}^{n}\frac{\allowbreak
z_{i}\phi (z_{i})C^{^{\prime \prime }}(\theta \Phi (z_{i}))}{C^{^{\prime
}}(\theta \Phi (z_{i}))} \\
&&-\frac{\theta }{\sigma }\sum_{i=1}^{n}\frac{z_{i}\phi (z_{i})\Phi
(z_{i})C^{^{\prime }}(\theta \Phi (z_{i}))C^{^{\prime \prime \prime
}}(\theta \Phi (z_{i})-z_{i}\phi (z_{i})\Phi (z_{i})\left( C^{^{\prime
\prime }}(\theta \Phi (z_{i})\right) ^{2}}{\left( C^{^{\prime }}(\theta \Phi
(z_{i}))\right) ^{2}}, \\
I_{\theta \theta } &=&-\frac{n}{\theta ^{2}}+\sum_{i=1}^{n}\frac{\Phi
^{2}(z_{i})C^{^{\prime \prime \prime }}(\theta \Phi (z_{i}))C^{^{\prime
}}(\theta \Phi (z_{i}))-\Phi ^{2}(z_{i})\left( C^{^{\prime \prime }}(\theta
\Phi (z_{i}))\right) ^{2}}{\left( C^{^{\prime }}(\theta \Phi (z_{i}))\right)
^{2}} \\
&&-\frac{nC^{^{\prime \prime }}(\theta )}{C(\theta )}+\frac{n\left(
C^{^{\prime }}(\theta )\right) ^{2}}{\left( C(\theta )\right) ^{2}}.
\end{eqnarray*}

The asymptotic distribution of $\sqrt{n}\left( \widehat{\Psi }-\Psi \right) $
is $N_{3}(0,I_{n}(\Psi )^{-1})$, that can be used to create approximate
intervals and confidence regions for the parameters and for the hazard and
the survival functions. An $100(1-\gamma )$ asymptotic confidence interval
for each parameter $\Psi _{r}$ is given by%
\begin{equation*}
ACI_{r}=\left( \widehat{\Psi }_{r}-Z_{\gamma /2}\sqrt{\widehat{I}^{rr}},%
\widehat{\Psi }_{r}+Z_{\gamma /2}\sqrt{\widehat{I}^{rr}}\right) ,
\end{equation*}%
where $\widehat{I}^{rr}$ is the $(r,r)$ diagonal element of $I_{n}(\widehat{%
\Psi })^{-1}$ for $r=1,2,3$ and $Z_{\gamma /2}$ is the quantile $1-\gamma /2$
of the standard normal distribution.

\section{EM-algorithm}

The EM algorithm is one such elaborate technique. The EM algorithm
is a general method of finding the maximum likelihood estimate of the
parameters of an underlying distribution from a given data set when the data
is incomplete or has missing values. There are two main applications of the
EM algorithm. The first occurs when the data indeed has missing values, due
to problems with or limitations of the observation process. The second
occurs when optimizing the likelihood function is analytically intractable
but when the likelihood function can be simplified by assuming the existence
of and values for additional but missing (or hidden) parameters.

We define a hypothetical complete-data distribution with a joint probability
density function in the form

\begin{equation*}
g(z,y;\Psi )=\frac{a_{z}\theta ^{z}}{\sigma C(\theta )}z\phi \left( \frac{%
y_{i}-\mu }{\sigma }\right) \Phi ^{z-1}\left( \frac{y_{i}-\mu }{\sigma }%
\right),
\end{equation*}%
where $ \sigma >0, \theta\,y\mathbf{\in }\mathbb{R}$ and $z\in N$. The
probability density function of $Z$ given $Y=y$ is given by%
\begin{equation*}
g(z\mid y)=\frac{g(z,y;\Psi )}{f(y)}=\frac{a_{z}\theta ^{z-1}z\Phi
^{z-1}\left( \frac{y_{i}-\mu }{\sigma }\right) }{C^{^{\prime }}(\theta \Phi
\left( \frac{y_{i}-\mu }{\sigma }\right) )}.
\end{equation*}%
After some simple calculation we have 
\begin{equation*}
E(Z\mid Y=y)=1+\frac{\theta \Phi \left( \frac{y_{i}-\mu }{\sigma }%
\right) C^{^{\prime \prime }}(\theta \Phi \left( \frac{y_{i}-\mu }{\sigma }%
\right) )}{C^{^{\prime }}(\theta \Phi \left( \frac{y_{i}-\mu }{\sigma }%
\right) )}.
\end{equation*}%
The complete-data log- likelihood has the form%
\begin{equation*}
l_{n}^{\ast }\ (\mathbf{y},\mathbf{z};\mu ,\sigma ,\theta )\varpropto \sum_{i=1}^{n}z_{i}\log
\theta -n\log \sigma -\frac{1}{2\sigma ^{2}}\sum_{i=1}^{n}\left( y_{i}-\mu
\right) ^{2}+\sum_{i=1}^{n}\left( z_{i}-1\right) \log \Phi \left( \frac{%
y_{i}-\mu }{\sigma }\right) -n\log C(\theta ).
\end{equation*}%
The components of the score function $U_{c}(\mathbf{y},\mathbf{z};\Psi )=\left( \frac{\partial l_{n}^{\ast }}{\partial \mu }, \frac{\partial l_{n}^{\ast }}{\partial \sigma }, \frac{\partial l_{n}^{\ast }}{\partial \theta }\right) $, are
\begin{eqnarray*}
\frac{\partial l_{n}^{\ast }}{\partial \mu } &=&\frac{1}{\sigma ^{2}}%
\sum_{i=1}^{n}\left( y_{i}-\mu \right) -\frac{1}{\sigma }\sum_{i=1}^{n}%
\left( z_{i}-1\right) \frac{\phi \left( \frac{y_{i}-\mu }{\sigma }\right) }{%
\Phi \left( \frac{y_{i}-\mu }{\sigma }\right) }, \\
\frac{\partial l_{n}^{\ast }}{\partial \sigma } &=&-\frac{n}{\sigma }+\frac{1%
}{\sigma ^{3}}\sum_{i=1}^{n}\left( y_{i}-\mu \right) ^{2}-\frac{1}{\sigma
^{2}}\sum_{i=1}^{n}\left( z_{i}-1\right) \frac{\left( y_{i}-\mu \right) \phi
\left( \frac{y_{i}-\mu }{\sigma }\right) }{\Phi \left( \frac{y_{i}-\mu }{%
\sigma }\right) }, \\
\frac{\partial l_{n}^{\ast }}{\partial \theta } &=&\frac{1}{\theta }%
\sum_{i=1}^{n}z_{i}-n\frac{C^{\prime }(\theta )}{C(\theta )}.
\end{eqnarray*}%
The maximum likelihood estimates can be obtained from the iterative
algorithm given by%
\begin{eqnarray*}
&&\frac{1}{\widehat{\sigma }^{(h)}}\sum_{i=1}^{n}\left( y_{i}-\widehat{\mu }%
^{(h+1)}\right) -\sum_{i=1}^{n}\left( \widehat{z}_{i}^{(h)}-1\right) \frac{%
\phi \left( \frac{y_{i}-\widehat{\mu }^{(h+1)}}{\widehat{\sigma }^{(h)}}%
\right) }{\Phi \left( \frac{y_{i}-\widehat{\mu }^{(h+1)}}{\widehat{\sigma }%
^{(h)}}\right) } =0 ,\\
&&\frac{1}{\left( \widehat{\sigma }^{(h+1)})\right) ^{2}}\sum_{i=1}^{n}\left(
y_{i}-\widehat{\mu }^{(h)}\right) ^{2}-\frac{1}{\widehat{\sigma }^{(h+1)}}%
\sum_{i=1}^{n}\left( \widehat{z}_{i}^{(h)}-1\right) \frac{\left( y_{i}-%
\widehat{\mu }^{(h)}\right) \phi \left( \frac{y_{i}-\widehat{\mu }^{(h)}}{%
\widehat{\sigma }^{(h+1)}}\right) }{\Phi \left( \frac{y_{i}-\widehat{\mu }%
^{(h)}}{\widehat{\sigma }^{(h+1)}}\right) }-n =0, \\
&&\widehat{\theta }^{(h+1)} =\frac{C(\widehat{\theta }^{(h+1)})}{%
nC^{^{\prime }}(\widehat{\theta }^{(h+1)})}\sum_{i=1}^{n}\widehat{z}%
_{i}^{(h)},
\end{eqnarray*}%

where $\widehat{\mu }^{(h+1)},$ $\widehat{\sigma }^{(h+1)}$ and \ $\widehat{%
\theta }^{(h+1)}$ are found numerically.  Here, for $i=1,...,n$, we have
that%
\begin{equation*}
\widehat{z}_{i}^{(h+1)}=1+\frac{\widehat{\theta }^{(h)}\Phi \left( \frac{%
y_{i}-\widehat{\mu }^{(h)}}{\widehat{\sigma }^{(h)}}\right) C^{^{\prime
\prime }}\left( \widehat{\theta }^{(h)}\Phi \left( \frac{y_{i}-\widehat{\mu }%
^{(h)}}{\widehat{\sigma }^{(h)}}\right) \right) }{C^{^{\prime }}(\widehat{%
\theta }^{(h)}\Phi \left( \frac{y_{i}-\widehat{\mu }^{(h)}}{\widehat{\sigma }%
^{(h)}}\right) )}.
\end{equation*}

\subsection{Evaluation of the standard errors from the EM-algorithm}

By using the results of Louis (1982) we obtain the standard errors
of the estimators from the EM-algorithm. The elements of the $3\times 3$
observed information matrix $I_{c}\left( \Psi; \mathbf{y},\mathbf{z}\right) =-\left[ 
\frac{\partial U_{C}\left(\mathbf{y},\mathbf{z}; \Psi \right) }{\partial \Psi }\right] $
are given by

\begin{eqnarray*}
\frac{\partial ^{2}l_{n}^{\ast }}{\partial \mu ^{2}} &=&\frac{n}{\sigma ^{2}}%
+\frac{1}{\sigma ^{2}}\sum_{i=1}^{n}\left( z_{i}-1\right) \frac{\left( \frac{%
y_{i}-\mu }{\sigma }\right) \phi \left( \frac{y_{i}-\mu }{\sigma }\right)
\Phi \left( \frac{y_{i}-\mu }{\sigma }\right) +\phi ^{2}\left( \frac{%
y_{i}-\mu }{\sigma }\right) }{\Phi ^{2}\left( \frac{y_{i}-\mu }{\sigma }%
\right) } ,\\
\frac{\partial ^{2}l_{n}^{\ast }}{\partial \mu \partial \sigma } &=&\frac{%
\partial ^{2}l_{n}^{\ast }}{\partial \sigma \partial \mu }=\frac{2}{\sigma
^{3}}\sum_{i=1}^{n}\left( y_{i}-\mu \right) -\frac{1}{\sigma ^{2}}%
\sum_{i=1}^{n}\left( z_{i}-1\right) \frac{\left( y_{i}-\mu \right) \phi
\left( \frac{y_{i}-\mu }{\sigma }\right) }{\Phi \left( \frac{y_{i}-\mu }{%
\sigma }\right) } \\
&&+\frac{1}{\sigma ^{2}}\sum_{i=1}^{n}\left( z_{i}-1\right) \frac{\left( 
\frac{y_{i}-\mu }{\sigma }\right) ^{2}\phi \left( \frac{y_{i}-\mu }{\sigma }%
\right) \Phi \left( \frac{y_{i}-\mu }{\sigma }\right) +\left( \frac{%
y_{i}-\mu }{\sigma }\right) \phi ^{2}\left( \frac{y_{i}-\mu }{\sigma }%
\right) }{\Phi ^{2}\left( \frac{y_{i}-\mu }{\sigma }\right) }, \\
\frac{\partial ^{2}l_{n}^{\ast }}{\partial \sigma ^{2}} &=&-\frac{n}{\sigma
^{2}}\allowbreak +\frac{3}{\sigma ^{4}}\sum_{i=1}^{n}\left( y_{i}-\mu
\right) ^{2}-\frac{2}{\sigma ^{3}}\sum_{i=1}^{n}\left( z_{i}-1\right) \frac{%
\left( y_{i}-\mu \right) \phi \left( \frac{y_{i}-\mu }{\sigma }\right) }{%
\Phi \left( \frac{y_{i}-\mu }{\sigma }\right) } \\
&&+\frac{1}{\sigma ^{3}}\sum_{i=1}^{n}\left( z_{i}-1\right) \left( y_{i}-\mu
\right) \frac{\left( \frac{y_{i}-\mu }{\sigma }\right) ^{2}\phi \left( \frac{%
y_{i}-\mu }{\sigma }\right) \Phi \left( \frac{y_{i}-\mu }{\sigma }\right)
-\left( \frac{y_{i}-\mu }{\sigma }\right) \phi ^{2}\left( \frac{y_{i}-\mu }{%
\sigma }\right) }{\Phi ^{2}\left( \frac{y_{i}-\mu }{\sigma }\right) }, \\
\frac{\partial ^{2}l_{n}^{\ast }}{\partial \theta ^{2}} &=&\frac{1}{\theta
^{2}}\sum_{i=1}^{n}z_{i}+n\frac{C^{\prime \prime }(\theta )C(\theta )-\left(
C^{\prime }(\theta )\right) ^{2}}{C^{2}(\theta )},\ \ \ \ \frac{\partial
^{2}l_{n}^{\ast }}{\partial \theta \partial \mu }=\ \frac{\partial
^{2}l_{n}^{\ast }}{\partial \mu \partial \theta }=\frac{\partial
^{2}l_{n}^{\ast }}{\partial \sigma \partial \theta }=\frac{\partial
^{2}l_{n}^{\ast }}{\partial \sigma \partial \theta }=0.
\end{eqnarray*}%
Taking the conditional expectation of $I_{c}\left(\Psi;\mathbf{y},\mathbf{z} \right) =-%
\left[ \frac{\partial U_{C}\left( \mathbf{y},\mathbf{z}; \Psi \right) }{\partial \Psi }%
\right] $ given $\mathbf{y}$, we obtain the $3\times 3$ matrix%
\begin{equation}\label{lc}
\mathit{l}_{c}\left( \Psi ;\mathbf{y},\mathbf{z}\right) =E\left( I_{c}\left( \Psi ;%
\mathbf{y},\mathbf{z}\right) \mid \mathbf{y}\right) =\left[ c_{ij}\right] 
\end{equation}
where
\begin{eqnarray*}
c_{11} &=&\frac{n}{\sigma ^{2}}+\frac{1}{\sigma ^{2}}\sum_{i=1}^{n}\left(
E(Z_{i}\mid y)-1\right) \left( \frac{y_{i}-\mu }{\sigma }\right) \frac{\phi
\left( \frac{y_{i}-\mu }{\sigma }\right) \Phi \left( \frac{y_{i}-\mu }{%
\sigma }\right) +\phi ^{2}\left( \frac{y_{i}-\mu }{\sigma }\right) }{\Phi
^{2}\left( \frac{y_{i}-\mu }{\sigma }\right) } \\
c_{21} &=&c_{12}=\frac{2}{\sigma ^{3}}\sum_{i=1}^{n}\left( y_{i}-\mu \right)
-\frac{1}{\sigma ^{2}}\sum_{i=1}^{n}\left( E(Z_{i}\mid y)-1\right) \frac{%
\left( y_{i}-\mu \right) \phi \left( \frac{y_{i}-\mu }{\sigma }\right) }{%
\Phi \left( \frac{y_{i}-\mu }{\sigma }\right) } \\
&&+\frac{1}{\sigma ^{2}}\sum_{i=1}^{n}\left( E(Z_{i}\mid y)-1\right) \frac{%
\left( \frac{y_{i}-\mu }{\sigma }\right) ^{2}\phi \left( \frac{y_{i}-\mu }{%
\sigma }\right) \Phi \left( \frac{y_{i}-\mu }{\sigma }\right) -\left( \frac{%
y_{i}-\mu }{\sigma }\right) \phi ^{2}\left( \frac{y_{i}-\mu }{\sigma }%
\right) }{\Phi ^{2}\left( \frac{y_{i}-\mu }{\sigma }\right) }, \\
c_{22} &=&-\frac{n}{\sigma ^{2}}\allowbreak +\frac{3}{\sigma ^{4}}%
\sum_{i=1}^{n}\left( y_{i}-\mu \right) ^{2}-\frac{2}{\sigma ^{3}}%
\sum_{i=1}^{n}\left( E(Z_{i}\mid y)-1\right) \frac{\left( y_{i}-\mu \right)
\phi \left( \frac{y_{i}-\mu }{\sigma }\right) }{\Phi \left( \frac{y_{i}-\mu 
}{\sigma }\right) } \\
&&+\frac{1}{\sigma ^{3}}\sum_{i=1}^{n}\left( E(Z_{i}\mid y)-1\right) \left(
y_{i}-\mu \right) \frac{\left( \frac{y_{i}-\mu }{\sigma }\right) ^{2}\phi
\left( \frac{y_{i}-\mu }{\sigma }\right) \Phi \left( \frac{y_{i}-\mu }{%
\sigma }\right) -\left( \frac{y_{i}-\mu }{\sigma }\right) \phi ^{2}\left( 
\frac{y_{i}-\mu }{\sigma }\right) }{\Phi ^{2}\left( \frac{y_{i}-\mu }{\sigma 
}\right) }, \\
c_{33} &=&\frac{1}{\theta ^{2}}\sum_{i=1}^{n}E(Z_{i}\mid y)+n\frac{C^{\prime
\prime }(\theta )C(\theta )-\left( C^{\prime }(\theta )\right) ^{2}}{%
C^{2}(\theta )},\ \ \ c_{13}=\ c_{31}=c_{23}=c_{32}=0,
\end{eqnarray*}%
and%
\begin{equation*}
E(Z_{i}\mid \mathbf{y})=1+\frac{\theta \Phi \left( \frac{y_{i}-\mu }{\sigma }\right)
C^{^{\prime \prime }}(\theta \Phi \left( \frac{y_{i}-\mu }{\sigma }\right) )%
}{C^{^{\prime }}(\theta \Phi \left( \frac{y_{i}-\mu }{\sigma }\right) )}.
\end{equation*}%
Moving now to the computation of $\mathit{l}_{m}\left( \Psi ;\mathbf{y}%
\right) \ $as%
\begin{equation}\label{lm}
\mathit{l}_{m}\left(\Psi;\mathbf{y}\right) =Var[U_{C}\left(\mathbf{y},\mathbf{z} ;\Psi
\right) \mid \mathbf{y}]=[v_{ij}],
\end{equation}%
where%
\begin{eqnarray*}
v_{11} &=&\frac{1}{\sigma ^{2}}\sum_{i=1}^{n}\frac{\phi ^{2}\left( \frac{%
y_{i}-\mu }{\sigma }\right) }{\Phi ^{2}\left( \frac{y_{i}-\mu }{\sigma }%
\right) }Var[Z_{i}\mid \mathbf{y}],\ \ \ \ \ \ \ v_{22}=\sum_{i=1}^{n}\left( \frac{%
\left( y_{i}-\mu \right) \phi \left( \frac{y_{i}-\mu }{\sigma }\right) }{%
\sigma ^{2}\Phi \left( \frac{y_{i}-\mu }{\sigma }\right) }\right)
^{2}Var[Z_{i}\mid \mathbf{y}], \\
v_{33} &=&\frac{1}{\theta ^{2}}\sum_{i=1}^{n}Var[Z_{i}\mid \mathbf{y}],\ \ \ \ \ \
v_{21}=v_{12}=\frac{1}{\sigma ^{3}}\sum_{i=1}^{n}\left( y_{i}-\mu \right)
\left( \frac{\phi \left( \frac{y_{i}-\mu }{\sigma }\right) }{\Phi \left( 
\frac{y_{i}-\mu }{\sigma }\right) }\right) ^{2}Var[Z_{i}\mid \mathbf{y}],\  \\
v_{13} &=&v_{31}=-\frac{1}{\sigma \theta }\sum_{i=1}^{n}\frac{\phi \left( 
\frac{y_{i}-\mu }{\sigma }\right) }{\Phi \left( \frac{y_{i}-\mu }{\sigma }%
\right) }\ \ Var[Z_{i}\mid \mathbf{y}],\ \ \ \  \\
v_{23} &=&v_{32}=-\frac{1}{\theta \sigma ^{2}}\sum_{i=1}^{n}\frac{\left(
y_{i}-\mu \right) \phi \left( \frac{y_{i}-\mu }{\sigma }\right) }{\Phi
\left( \frac{y_{i}-\mu }{\sigma }\right) }Var[Z_{i}\mid \mathbf{y}].\ \ \ \ \ \ \ \ \
\ \ 
\end{eqnarray*}%
and%
\begin{eqnarray*}
Var\left[Z_{i}\mid \mathbf{y}\right]\ &=&E\left( Z_{i}^{2}\mid \mathbf{y}\right) -\left( E\left(
Z_{i}\mid \mathbf{y}\right) \right) ^{2} \\
&=&\frac{1}{C^{^{\prime }}(\theta _{\ast })}\sum_{z=1}^{n}a_{z}z^{3}\theta
_{\ast }^{z-1}-\frac{\left( C^{^{\prime }}(\theta _{\ast })+\theta _{\ast
}C^{^{^{\prime \prime }}}(\theta _{\ast })\right) ^{2}}{\left( C^{^{\prime
}}(\theta _{\ast })\right) ^{2}} \\
&=&\frac{1}{C^{^{\prime }}(\theta _{\ast })}[\theta _{\ast
}^{2}C^{^{^{\prime \prime \prime }}}(\theta _{\ast })+C^{^{\prime }}(\theta
_{\ast })+3\theta _{\ast }C^{^{^{\prime \prime }}}(\theta _{\ast })]-\frac{%
[C^{^{^{\prime }}}(\theta _{\ast })+\theta _{\ast }C^{^{^{\prime \prime
}}}(\theta _{\ast })]^{2}}{\left( C^{^{\prime }}(\theta _{\ast })\right) ^{2}%
},
\end{eqnarray*}%
in which $\theta _{\ast }=\theta \Phi \left( \frac{y_{i}-\mu }{\sigma }%
\right) .$\\
and applying \eqref{lc} and \eqref{lm}, we obtain the observed information as%
\begin{equation*}
I\left( \widehat{\Psi };\mathbf{y}\right) =\mathit{l}_{c}\left( \widehat{%
\Psi };\mathbf{y}\right) -\mathit{l}_{m}\left( \widehat{\Psi };\mathbf{y}%
\right).
\end{equation*}%
The standard errors of the MLEs of the EM-algorithm are the square root of
the diagonal elements of the $I\left( \widehat{\Psi };\mathbf{y}\right) .$
\section{Simulation}

This section provides the results of simulation study. Simulations have been
performed in order to investigate the proposed estimator of, $\mu $, $\sigma 
$, $\theta $ of the proposed MLE method. We simulate 1000 times 
under the NG distribution with three different sets of parameters and  sample sizes n =50, 100, 200 and 300. For each
sample size, we compute the MLEs by EM-method. We also compute the standard
error of the MLEs of the EM-algorithm determined though the Fisher
information matrix. The simulated values of $se(\widehat{\mu })$, $se(%
\widehat{\sigma }),$ $se(\widehat{\theta }),$ $Cov(\widehat{\mu },\widehat{%
\sigma }),$ $Cov(\widehat{\mu },\widehat{\theta })$ and $Cov(\widehat{\sigma 
},\widehat{\theta })$ obtained by averaging the corresponding values of the
observed information matrices, are computed. The results for the NP
distribution are reported in Tables \ref{tablesi}. Some of the points are quite
clear from the simulation results: ($i$) Convergence has been achieved in
all cases and this emphasizes the numerical stability of the EM-algorithm. ($%
ii$) The differences between the average estimates and the true values are
almost small. ($iii$) These results suggest that the EM estimates have
performed consistently. ($iv$) As the sample size increases, the standard
errors of the MLEs decrease. ($v$) Additionally, the standard errors of MLEs
of the EM-algorithm obtained from the observed information matrix are quite
close to the simulated ones for large values of $n$.

\begin{table}
\begin{footnotesize}
\begin{center}
\caption{The averages of the 1000 MLE's, mean of the simulated
standard errors and mean of the standard errors of EM estimators
obtained using observed information matrix of the NP distribution.
} \label{tablesi}
\begin{tabular}{|c|c|ccc|ccc|ccc|} \hline
 &   & \multicolumn{3}{|c|}{Average estimators} & \multicolumn{3}{|c|}{S.td} & \multicolumn{3}{|c|}{Cov} \\ \hline
 $n$ & $(\mu ,\sigma ,\theta )$ & $\widehat{\mu }$  & $\widehat{\sigma }$  & $\widehat{\theta }$  & $se(\widehat{\mu })$
 & $se(\widehat{\sigma })$  & $se(\widehat{\theta })$  & $cov(\widehat{\mu } ,\widehat{\sigma }))$  & $cov(\widehat{\mu },\widehat{\theta } )$  & $cov(\widehat{\sigma },\widehat{\theta })$  \\ \hline
30&(0,1.0,1.0)&0.39&1.06&-8.92&1.05&0.17&100.67&0.10&-40.46&-5.90 \\
&(0,1.0,0.5)&0.17&1.02&-1.95&0.9&0.14&21.66&0.00&-7.17&-0.59 \\
&(0,1.0,0.8)&-0.17&1.03&0.57&1.02&0.17&0.81&-0.14&-0.46&0.03 \\
  \hline
100&(0,1.0,1.0)&0.16&1.02&-1.74&0.7&0.1&10.78&0.04&-43.4&-0.60 \\
&(0,1.0,0.5)&0.02&1.01&0.11&0.62&0.1&1.75&-0.02&-0.69&-0.03 \\
&(0,1.0,0.8)&-0.12&1.03&0.69&0.79&0.14&0.46&-0.09&-0.19&0.02 \\
  \hline

200&(0,1.0,1.0)&0.07&1.01&-0.52&0.47&0.00&2.90&0.46&-14.1&-0.14 \\
&(0,1.0,0.5)&0.00&1.00&0.39&0.37&0.00&0.56&-0.01&-0.15&0.01 \\
&(0,1.0,0.8)&-0.08&1.02&0.75&0.59&0.1&0.17&-0.05&-0.08&0.01 \\  
\hline

300&(0,1.0,1.0)&0.03&1.00&-0.27&0.39&0.00&2.33&0.01&-0.73&-0.09 \\
&(0,1.0,0.5)&0.00&1.00&0.40&0.3&0.00&0.9&0.00&-0.16&0.00 \\
&(0,1.0,0.8)&-0.06&1.01&0.77&0.48&0.1&0.17&-0.03&-0.06&0.01 \\
\hline
%

\end{tabular}
\end{center}
\end{footnotesize}
\end{table}

\section{Application}
In this section, we try to illustrate the better performance of the proposed model. For this end, we fit NG, NP and NL models to two real data sets. We also fit the Azzalini's skew-normal (ASN) and normal distributions to make a comparison with the NPS models.
The first data concerning the heights (in centimeters) of 100 Australian athletes. The data have been previously analyzed in Cook and Weisberg and are available for
download at \textit{http://azzalini.stat.unipd.it/SN/index.html}. We estimate parameters by numerically maximizing the likelihood function. The MLEs of the parameters, the maximized loglikelihood,
 the AIC (Akaike Information Criterion) and BIC (Bayesian Information Criterion) for the NG, NP, NL, normal and azzalini's skew-normal models are given in Table \ref{tab data 1}.\\
\begin{table}
\centering
\caption{Parameter estimates, AIC and BIC for AIS data.}
\begin{tabular}{|l|llll|}
\hline
Dist.&Parameter estimates&$−\log(L)$&AIC&BIC\\
\hline
NG&$\widehat{\mu }$=136.001, $\widehat{\sigma }$=13.642, $\widehat{\theta }$=0.998&348.376&702.752&710.567\\
NP&$\widehat{\mu }$=167.106, $\widehat{\sigma }$=9.208, $\widehat{\theta }$=3.398&349.145&704.291&712.106\\
NL&$\widehat{\mu }$=169.353, $\widehat{\sigma }$=7.947, $\widehat{\theta }$=0.897&350.872&707.745&715.560\\
Normal&$\widehat{\mu }$=174.594, $\widehat{\sigma }$=8.209&352.319&708.635& 713.846\\
ASN&$\widehat{\mu }$=170.320, $\widehat{\sigma }$=8.002, $\widehat{\theta }$=0.0016&352.032&710.636&718.451\\
\hline
\end{tabular}\label{tab data 1}
\end{table}
As is well known, a model with a minimum AIC value is to be preferred. Therefor NG distribution provides a better fit to this data set than the other distributions and hence could be chosen as the best distribution. Also this conclusion is confirmed from the plots of the densities functions in Figure \ref{AIS}.
\begin{figure}
\centering
\includegraphics[scale=0.40]{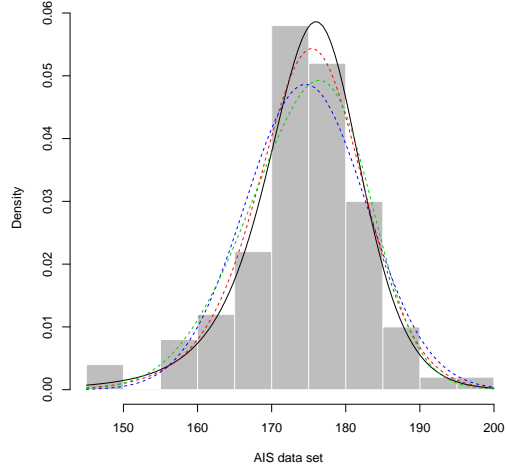}
\caption[]{Histogram of heights of 100 Australian athletes.The lines represent distributions
fitted using maximum likelihood estimation: NG (Black), NP (Red), NL(Green) and ASN (Blue)}\label{AIS}.
\end{figure}

The second data represent the Oits IQ Scores for 52 non-White males hired by a large insurance company in 1971 given in
Roberts (1988).
Table \ref{tab data 2} gives the MLEs of the parameters, the maximized log-likelihood, the AIC and BIC for the NG, NP, NL, Normal and ASN models for the second data set.
\begin{table}
\centering
\caption{Parameter estimates, AIC and BIC for OTIS IQ scores data.}
\begin{tabular}{|l|llll|}
\hline
Dist.&Parameter estimates&−log(L)&AIC&BIC\\
\hline
NG&$\widehat{\mu }$=112.875, $\widehat{\sigma }$=182.313, $\widehat{\theta }$=-2.989&182.313&370.628&376.479\\
NP&$\widehat{\mu }$=106.263, $\widehat{\sigma }$=8.227, $\widehat{\theta }$=0.0000002&182.313&372.850&376.479\\
NL&$\widehat{\mu }$=106.308, $\widehat{\sigma }$=7.947, $\widehat{\theta }$=0.0000002&183.433&372.867& 378.719\\
Normal&$\widehat{\mu }$=106.654, $\widehat{\sigma }$=8.230&183.387&370.774&   374.676\\
ASN&$\widehat{\mu }$=98.790, $\widehat{\sigma }$=11.380, $\widehat{\theta }$=1.710&182.436&370.872& 376.726\\
\hline
\end{tabular}\label{tab data 2}
\end{table}

The results for this data set show that the NG distributions yield the best fit among the NG, NL, normal and ASN distributions. Also the plots of the densities function
in Fiure \ref{OTIS} confirmed this conclusion.
\\
\begin{figure}
\centering
\includegraphics[scale=0.40]{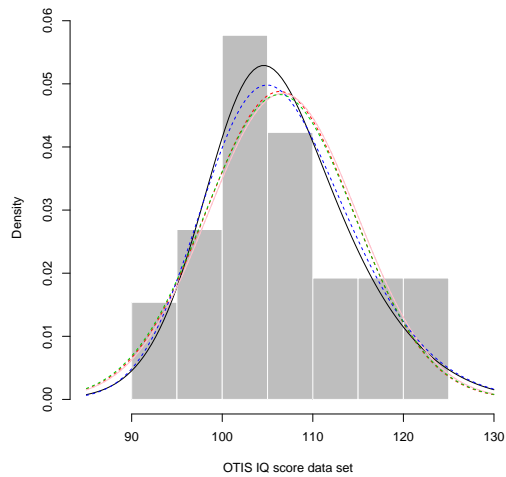}
\caption[]{Histogram of heights of 100 Australian athletes.The lines represent distributions
fitted using maximum likelihood estimation: NG (Black), NP (Red), NL(Green), normal(pink) and ASN (Blue)}\label{OTIS}.
\end{figure}\\
\section{Conclusion}
In this paper we introduce a new three-parameter class of distributions called the normal power series distributions (NPS), which is an alternative to the Azzalini skew-normal distribution for fitting skewed data. The NPS distributions contain the NG, NP, NB and NL distributions as special cases. We obtain closed form expressions for the moments. The estimation of the unknown
parameters of the proposed distribution is approached by the EM-algorithm. Finally, we fitted NPS models to two real data sets to show the potential of the new proposed class.

\end{document}